\documentclass[11pt]{article}
\usepackage[utf8]{inputenc}
\usepackage{amsfonts,latexsym,amsthm,amssymb,amsmath,amscd,euscript,mathrsfs,graphicx}
\usepackage{framed}
\usepackage{fullpage}
\usepackage{color}
\usepackage{enumitem}
\usepackage{authblk}
\usepackage{mathtools}
\usepackage[hidelinks]{hyperref}
\usepackage{thm-restate}

\usepackage{algorithm}
\usepackage[noend]{algpseudocode}
\usepackage{todonotes}

\newtheorem{theorem}{Theorem}[section]
\newtheorem{proposition}[theorem]{Proposition}
\newtheorem{lemma}[theorem]{Lemma}
\newtheorem{corollary}[theorem]{Corollary}

\theoremstyle{definition}
\newtheorem{definition}[theorem]{Definition}

\theoremstyle{remark}
\newtheorem*{remark}{Remark}

\usepackage[font=itshape]{quoting}

\newcommand{\BN}{\mathbb N}

\newcommand{\BR}{\mathbb R}
\newcommand{\BZ}{\mathbb Z}
\newcommand{\eps}{\varepsilon}
\newcommand{\poly}{\operatorname{poly}}
\newcommand{\GREEDY}{\operatorname{GREEDY}}
\newcommand{\OPT}{\operatorname{OPT}}
\newcommand{\defnn}[1]{\textbf{\emph{#1}}}

\renewcommand{\paragraph}{\medskip \textbf}

\title{Optimal Time-Backlog Tradeoffs for the Variable-Processor Cup Game}
\author{William Kuszmaul\thanks{Massachusetts Institute of Technology. Email: \texttt{kuszmaul@mit.edu}} , Shyam Narayanan\thanks{Massachusetts Institute of Technology. Email: \texttt{shyamsn@mit.edu}.}}

\begin{document}

\maketitle

\begin{abstract}
The \emph{$ p$-processor cup game} is a classic and widely studied scheduling problem that captures the setting in which a $p$-processor machine must assign tasks to processors over time in order to ensure that no individual task ever falls too far behind. The problem is formalized as a multi-round game in which two players, a filler (who assigns work to tasks) and an emptier (who schedules tasks) compete. The emptier's goal is to minimize backlog, which is the maximum amount of outstanding work for any task.

Recently, Kuszmaul and Westover (ITCS, 2021) proposed the \emph{variable-processor cup game}, which considers the same problem, except that the amount of resources available to the players (i.e., the number $p$ of processors) fluctuates between rounds of the game. They showed that this seemingly small modification fundamentally changes the dynamics of the game: whereas the optimal backlog in the fixed $p$-processor game is $\Theta(\log n)$, independent of $p$, the optimal backlog in the variable-processor game is $\Theta(n)$. The latter result was only known to apply to games with \emph{exponentially many} rounds, however, and it has remained an open question what the optimal tradeoff between time and backlog is for shorter games. 

This paper establishes a tight trade-off curve between time and backlog in the variable-processor cup game. We show that, for a 
game consisting of $t$ rounds, the optimal backlog is $\Theta (b (t)) $ where
$$
b(t) = \begin{cases}
t &\text { if } t\le \log n\\
t ^ { 1/3 } \log ^ { 2/3 } \left(\frac{n ^ 3}{t} + 1\right) &\text { if }\log n < t \le n ^ 3\\
n &\text { if } n ^ 3 < t.
\end{cases}
$$
An important consequence is that the optimal backlog is $\Theta(n)$ if and only if $t \ge \Omega(n^3)$. 
Our techniques also allow for us to resolve several other open questions concerning how the variable-processor cup game behaves in beyond-worst-case-analysis settings. 
\end{abstract}

\section{Introduction} \label{sec:intro}

\textbf{The classical $p$-processor cup game.} 
The $p$ processor cup game captures the general problem in which there are some number $n$ of tasks competing for a smaller number $p$ of processors~\cite{BaruahCoPl96,GkasieniecKl17,BaruahGe95,LitmanMo11,LitmanMo05,MoirRa99,BarNi02,GuanYi12,Liu69, LiuLa73,DietzRa91, BenderFaKu19, Kuszmaul20, AdlerBeFr03, DietzSl87, LitmanMo09, im2021matroid}. 
A scheduler must assign tasks to processors over time in order to ensure that no individual task ever falls too far behind.

Formally, this is captured as a game with $n$ cups, each capable of holding an arbitrarily large amount of water, and two competing players, a filler and an emptier.
 In each round of the game, the filler distributes $p$ new units of water into the cups, placing at most 1 unit of water into any particular cup. 
 The emptier then selects $p$ distinct cups and removes up to 1 unit of water from each of them. Note that, whereas the filler may place their $p$ units of water in fractional amounts across arbitrarily many cups, the emptier can only choose $p$ cups per step to empty from. The emptier's goal is to minimize the backlog of the system, which is the amount of water in the fullest cup.

If one views the cup game as a scheduling problem, then the cups represent tasks, the water represents work, the filler represents an adversary that determines when work arrives,
and the emptier represents a scheduler that can select $p$ tasks to run on a given time step (we will use the terms ``round'' and ``time step'' interchangeably). Although we will primarily be interested in the cup game as a scheduling problem~\cite{BaruahCoPl96,GkasieniecKl17,BaruahGe95,LitmanMo11,LitmanMo05,MoirRa99,BarNi02,GuanYi12,Liu69, LiuLa73, AdlerBeFr03, LitmanMo09, DietzRa91}, it has also found
applications to many other problems (e.g. deamortization of data structures~\cite{AmirFaId95,DietzRa91,DietzSl87,AmirFr14,Mortensen03,GoodrichPa13,FischerGa15,Kopelowitz12,BenderDaFa20}, network-switch buffer management~\cite{Goldwasser10,AzarLi06,RosenblumGoTa04,Gail93}, quality-of-service guarantees~\cite{BaruahCoPl96,AdlerBeFr03,LitmanMo09}, etc.).  

Beginning in the late 1960s, much of the early work on the $p$-processor cup game focused on the fixed-rate version of the game,
in which the filler's behavior is the same at every round~\cite{BaruahCoPl96,GkasieniecKl17,BaruahGe95,LitmanMo11,LitmanMo05,MoirRa99,BarNi02,GuanYi12,Liu69, LiuLa73}. In this version of the game, it is possible for the emptier to achieve a backlog
of $O(1)$, both in the single-processor cup game (i.e., $p = 1$) \cite{Liu69, LiuLa73} and in the multi-processor cup game (i.e., $p > 1$) \cite{BaruahCoPl96}.
In recent decades, much of the research has shifted to focus on the non-fixed-rate version of the game, in which the filler is an adaptive adversary that can change their behavior from step to step~\cite{AdlerBeFr03, DietzSl87, Kuszmaul20, LitmanMo09, Bar-NoyFrLa02, FleischerKo04, DamaschkeZh05}.
In this setting, it is possible for the emptier to achieve backlog $O (\log n) $ \cite{AdlerBeFr03, DietzSl87, Kuszmaul20},
and this is known to be asymptotically optimal for all $p \le n - \sqrt{n}$ \cite{Kuszmaul20}. There is also a long history of researchers applying
techniques from beyond-worst-case analysis to cup games, e.g., resource augmentation~\cite{BenderFaKu19, Kuszmaul20, LitmanMo09, DietzRa91}, smoothed analysis~\cite{Kuszmaul20, BenderFaKu19}, adversary restrictions~\cite{BenderFaKu19, Kuszmaul20, BenderKu20, tail, DietzRa91}, semi-clairvoyance~\cite{LitmanMo09}, etc. 

A repeating theme in these directions of work has been the relative difficulty of analyzing the multi-processor case in comparison to the
single-processor case. As Liu discussed in his seminal 1969 paper~\cite{Liu69}, and as many later authors have subsequently reiterated~\cite{LitmanMo09,BaruahCoPl96, GuanYi12, BenderFaKu19, Kuszmaul20}, the difficulty of the multi-processor case stems from the fact that the emptier must remove water from $p$ \emph{distinct} cups,
even if the vast majority of the water is in a smaller number of cups. For both the fixed-rate and the non-fixed-rate games,
the optimal backlog in the multi-processor version of the game~\cite{BaruahCoPl96, Kuszmaul20} was proven decades after the corresponding result for the single-processor game was first shown~\cite{Liu69, DietzSl87}.

\paragraph{The variable-processor cup game.} Recent work by Kuszmaul and Westover~\cite{variable} has considered the question of what happens if the parameter $p$ is permitted to change over time, with the filler
adaptively determining both what value of $p$ will be used at each round and where the $p$ new units of water are placed. The resulting game, which is known as the variable-processor cup game,
captures settings in which the amounts of resources available to the players fluctuate over time.\footnote{As discussed in~\cite{variable}, there is no fundamental reason
why the amount $p$ of resources should fluctuate in the same way for the filler as it does for the emptier over time. However, by assuming that $p$ is always the same for both players,
one ensures that there is a fair playing field: neither player has an advantage over the other in terms of their resources.}

The problem of what to do when computing resources fluctuate has received increasing attention
in recent years due to the proliferation of shared-computing systems in which multiple users and virtual operating systems simultaneously run
on a single physical multi-core machine; the fact that the machine is shared means that
the amount of resources (eg., cache, processors, memory bandwidth, etc.) available to each user is constantly changing, depending on the current demands
of other users. This phenomenon has led researchers to revisit problems in which computing resources have traditionally been viewed as static~\cite{CA1, CA2, CA3, CA4, CA5, variable}.

Intuitively, the variable-processor cup game would seem to be relatively similar to its classical $p$-processor counterpart.
Indeed, the backlog in the $p$-processor cup game is $O (\log n) $ regardless of the value of the parameter $p$, suggesting that the same should be true
if $p $ is permitted to vary. The central result of~\cite{variable} is that this intuition turns out to be completely wrong: given sufficiently many time steps,
the filler can actually force a backlog of $\Omega(n)$ in the variable-processor cup game, and this backlog is asymptotically optimal.

In order to achieve the backlog of $\Omega(n)$, the authors~\cite{variable} construct a strategy for the filler in which the number of processors $p$ follows a recursive ``fractal-like'' pattern.
The recursive structure requires a relatively large number of time steps to complete---to achieve the full backlog of $\Omega(n)$,
the construction requires exponentially many time steps.

The unexpectedly large backlog
prompts several questions. The main open question is the problem of determining the optimal trade-off between backlog and time in the variable-processor cup game,
and, in particular, what the optimal backlog is in games of polynomial lengths. 
In this setting, it is not even known whether the filler can achieve polynomial backlog---the best known filler strategy in this case~\cite{variable} achieves backlog $2^{O(\sqrt{\log n})}$.
Understanding the optimal trade-off between time and backlog in polynomial-length games is especially important since, for instance, one may have a bound on the number of rounds in a given scheduling application, which may allow for a much better guarantee on the backlog.

The authors of~\cite{variable} also raise the question of whether smaller bounds on backlog can be achieved via beyond-worst-case analysis. Based on their results, they propose two directions, in particular, that seem promising. One is to place an additional restriction on the filler that $p$ can only change at a certain rate;
this would thwart the recursive structure of their lower bound construction which changes $p$ dramatically between levels of recursion. The other is to consider the use of resource augmentation,
meaning that the emptier is allowed to remove slightly more water in each time step than the filler is permitted to place into the cups. This direction seems promising due to the
large amount of time required by the filling strategy of~\cite{variable}, since over such a large amount of time, the resource augmentation would potentially offer a large advantage to the emptier. 

\paragraph {Our results.} 
The main result of this paper is a tight trade-off curve between time and backlog in the variable-processor cup game. We show that, for a game consisting of $t$ rounds,
the optimal backlog is $\Theta (b (t)) $ where
\begin{equation} \label{trade-off-curve}
b(t) = \begin{cases}
t &\text { if } t\le \log n\\
t ^ { 1/3 } \log ^ { 2/3 } \left(\frac{n ^ 3}{t} + 1\right) &\text { if }\log n < t \le n ^ 3\\
n &\text { if } n ^ 3 < t.
\end{cases}
\end{equation}
By optimal, we mean that there exists an emptying strategy such that no filler can achieve backlog greater than $\Omega(b(t))$ after $t$ rounds, and there exists a filling strategy such that no emptier can achieve backlog less than $O(b(t))$ after $t$ rounds.
In the case of the emptying strategy, we show that this tradeoff curve is achieved by the greedy emptying algorithm (i.e., always empty from the fullest cups). 

Equation \eqref{trade-off-curve} comes with several interesting takeaways. The first is that in short games, of length $o(n^3)$, the emptier \emph{can} achieve sub-linear backlog---prior to this result, it remained open whether the emptier could even ensure $o(n)$ backlog for only $n$ rounds. The second is that in games of size $\Omega(n^3)$, the optimal backlog is $\Theta(n)$---this resolves the open question of~\cite{variable} as to whether
linear backlog can be achieved in polynomial time. The third is that the optimal tradeoff between backlog and time has a somewhat unexpected polylogarithmic low-order term, which disappears only when $t$ grows to be $\Omega(n^3)$. 

By examining the inverse of the function $b(t)$, another way to think about \eqref{trade-off-curve} is that, for any quantity $b \le n$, the amount of time $t(b)$ needed for an optimal filler to force a backlog of $\Omega(b)$ against an optimal emptier is 
\begin{equation}
t(b) = \Theta\left(b + \frac{b^3}{\log^2 \frac{n}{b}}\right),\label{eq:tradeoff2}
\end{equation}
and that backlogs $b = \omega(n)$ are not achievable by the filler (the latter fact, of course, is already known due to~\cite{variable}).


The second contribution of this paper is to analyze the variable-processor cup game under two forms of beyond-worst-case analysis,
each of which resolves an open question posed by~\cite{variable}.
We begin by considering the setting in which the rate at which the filler can change $p $ is severely limited: 
$p$ is permitted to change by most
$\pm 1$ per time step, and the filler can only change $p $ every $\poly (n)$ time steps.
Remarkably, this has no effect on the optimal backlog, and the filler can still force a backlog of $n/2 $ in polynomial time.
Next, we consider the setting in which the emptier has $\eps > 0$ resource augmentation, meaning that, in each time step, 
the emptier is permitted to remove up to $1+\eps $ (rather than 1) units of water from each of $p $ cups.
This has a dramatic effect on the optimal backlog, reducing it to $O(\frac{1}{\eps} \log n)$, which asymptotically matches 
the optimal backlog of the standard $p $-processor cup game when $\eps = \Omega(1)$. 

\paragraph{Techniques and paper outline. }
We formally describe the necessary preliminaries in Section \ref{sec:prelim}, and provide a detailed overview of the technical ideas in the paper in Section \ref{sec:technical}.

In Section \ref{sec:greedy_optimal}, we prove a result that ends up being useful in many of the later sections, namely that the greedy emptying algorithm is actually the \emph{exact} optimal online algorithm, and that this holds no matter what the starting state of the cups is. Interestingly, our proof {of greedy emptying being the optimal emptying strategy} also applies to the {fixed} $p $-processor cup game, for which the result was also not previously known: {in this setting, the greedy emptying algorithm was previously only known to be \emph{asymptotically} optimal~\cite{Kuszmaul20}, and this was only known for the starting state with all empty fills}.
We remark that, although the fact that greedy emptying is optimal is certainly intuitive, it actually isn't true for every variant of the cup game;
notably, the greedy emptying algorithm isn't even asymptotically optimal for the fixed-rate version of the game~\cite{AdlerBeFr03}.

In Section \ref{sec:lower}, we construct an asymptotically optimal strategy for the filler. The strategy achieves the bound in \eqref{trade-off-curve} no matter what strategy the emptier follows,
and the strategy also easily generalizes the case where the filler is restricted to change $p$ at a slow rate. 
This implies that the optimal backlog $b(t)$ is \emph{at least} that in \eqref{trade-off-curve}.
While the filling algorithm works against any emptier, we focus on the filler working against the greedy emptier, which we know is optimal.
We start with a warm-up (subsection \ref{subsec:lower_warmup}) proving that one can establish $\Omega(n)$ backlog in $O(n^3)$ time, and we then generalize our techniques to apply to games of arbitrary lengths.

The remainder of the paper has been deferred to the Appendix.
In Section \ref{sec:upper}, we turn our attention to proving a tight upper bound for the maximum backlog against a greedy emptier, that is, we show that $b(t)$ is \emph{at most} that in \eqref{trade-off-curve}.
We begin by creating a variation of the variable-processor cup game, which we call the \defnn{stone game}, in which the filler's behavior is limited in a certain combinatorially natural way. We analyze the maximum backlog of any filling strategy for the stone game by devising two potential functions and comparing their growth rates; this allows us to establish that $\Omega(n^3)$ time steps are needed to achieve backlog $\Omega(n)$ in the stone game.
We then tighten the bound on backlog when there are fewer time steps by partitioning the cups into levels and arguing that a constant fraction
of the levels interact especially nicely with the potential functions; this yields \eqref{trade-off-curve} for the stone game.
Finally, we show that, if the emptier behaves greedily, then this stone game encapsulates the main problem, that is, is always advantageous to the filler in the variable-processor cup game to act as though they are in the stone game. Thus we can transfer out bounds on the stone game into bounds on the variable-processor cup game.

Section \ref{sec:upper} is, in our opinion, the most technically involved section in our work. Perhaps the most interesting mathematical contribution in this section is to analyze the furthest backlog in the stone game by constructing two different potential functions and comparing their relative growth rates. While individual potential functions have been used to analyze cup games~\cite{BenderFaKu19}, this is the first time that comparing two potential functions has been applied to cup games.

Finally, in Section \ref{sec:resource_augmentation}, we give a very simple analysis of the variable-processor cup game in the presence of resource augmentation.
Our argument is non-constructive, employing the probabilistic method in order to show that there \emph{exists}
an emptying algorithm that achieves backlog $O(\frac{1}{\eps} \log n)$ (interestingly, the same argument 
also gives a nontrivial bound of $O(\sqrt{t \log n})$ in the resource-augmentation-free setting). 
Since the greedy emptying algorithm is optimal, it must also achieve the same bound.
To the best of our knowledge, this is the first example of the probabilistic method being used to analyze cup games.

\section{Preliminaries} \label{sec:prelim}

We first formally define the \defnn{variable-processor cup game}. In this game, there are $n$ cups of real-valued fills $x_1, \dots, x_n$, all starting at $0$, and two \emph{adaptive} players, a filler and an emptier. At each round, the filler chooses an integer $1 \le p \le n,$ chooses real numbers $a_1, \dots, a_n$ such that $0 \le a_i \le 1$ for all $i$ and $\sum_{i = 1}^{n} a_i = p,$ and replaces $x_i$ with $x_i' = x_i+a_i$ for all $1 \le i \le n$. The emptier then chooses a set $S \subset [n]$ of size $p$, and for each $i \in S$, replaces $x_i'$ with $\max(0, x_i'-1)$ but does not change $x_i'$ for $i \not\in S$. A single \defnn{round} (which we also call \defnn{time step}) consists of both the filler's and emptier's moves. We define the \defnn{state} of the cups at a fixed round to be the current sequence $\{x_1, x_2, \dots, x_n\}$ of the values of cups. Since the filler and emptier are adaptive, we note that two states are equivalent if the sequences are equal up to permutation.
We will also say that the \defnn{fills} are $x_1, x_2, \dots, x_n$ (we think of $x_i$ as the fill of the $i^{\text{th}}$ cup). Finally, for any state $X = \{x_1, \dots, x_n\}$, we define the \defnn{backlog} of $X$ as the maximum fill, or $\max_{1 \le i \le n} x_i$. The goal of the filler is to maximize the backlog after $t$ rounds for some fixed $t$, whereas the goal of the emptier is to minimize it. We also define the \defnn{variable-processor cup game with $\eps$ resource augmentation} as the same as the variable-processor cup game, except that the emptier, for each $i \in S$, replaces $x_i'$ with $\max(0, x_i'-(1+\eps))$.

Next, we define the \defnn{negative-fill variable-processor cup game} as the same as the variable-processor cup game, except that the emptier, for each $i \in S$, replaces $x_i'$ with $x_i'-1$. This may mean that some of the fills in a state become negative, which is allowed.  We analogously define round, state, fills, and backlog (note that the backlog is $\max_{1 \le i \le n} x_i$, not $\max_{1 \le i \le n} |x_i|$).

Unless explicitly stated otherwise, the \emph{standard} game refers to the variable-processor cup game, and the \emph{negative-fill} game refers to the negative-fill variable-processor cup game. We will explicitly state whenever we talk about the fixed $p$-processor cup game (where $p = \sum_{i = 1}^{n} a_i$ is fixed for every round), or the variable-processor cup game with $\eps > 0$ resource augmentation.


We conclude this section by briefly commenting on the relationship between the standard game and the negative-fill game. It is easy to see that the optimal backlog in the standard game is at least as large as the optimal backlog in the negative-fill game. What is less obvious, but worth noting, is that the optimal backlogs in the two games are actually asymptotically \emph{equal}. We provide a proof of this in Appendix \ref{sec:appendix_less_ez}. Thus, in general, either version of the game is equally valid (although the only place where we will take advantage of this in any nontrivial way will be Section \ref{subsec:upper_finish}).

\section{Technical Overview} \label{sec:technical}

In this section, we provide a technical overview for both the lower and upper bound for the variable-processor cup game, as well as the upper bound for the variable-processor cup game with resource augmentation.

\subsection{Overview of the Lower Bound on Backlog} \label{subsec:lower_overview}

In Section \ref{sec:lower}, we provide a lower bound on the backlog that the filler can achieve over $t$ rounds. For now, let us assume that we are playing the negative-fill game and that the emptier is greedy, i.e., at each time step, if the filler fills $p$ units of water, then the emptier empties $1$ unit from the $p$ fullest cups. We shall remove these assumptions at the end of the subsection.

\paragraph{Achieving backlog $\Omega(n)$ in $n^3$ steps.} Suppose at some time step, the fills are $x_1 \ge x_2 \ge \cdots \ge x_n$. Consider any subsequence $x_i, \ldots, x_j$ of even length, such that $x_i = x_{i + 1} = \cdots = x_j = k$ for some value $k$. We claim that it is possible for the filler to place water into the cups so that, once the greedy emptier has performed their move, the net change to the fills will simply be that: half of the cups $i, \ldots, j$ now have fills $k + \frac{1}{2}$ and half of the cups $i, \ldots, j$ now have fills $k - \frac{1}{2}$. Indeed, to achieve this, the filler can simply set $p = (i-1) + (j-i+1) / 2$, place $1$ unit of water into each of cups $1, 2, \ldots, i - 1$ (this water will promptly be removed by the emptier), and place $1/2$ units of water into each of cups $i, \ldots, j$ (the emptier will empty from half of these cups). 

Why is this type of move good for the filler? Consider the potential function measuring the sum of squares of the fills, i.e., $\Phi(x_1, \dots, x_n) := \sum_{i = 1}^{n} x_i^2$. If we let $q = \frac{j-i+1}{2}$, so that $q$ of the fills go up by $1/2$ and $q$ of the fills go down by $1/2$, then it is not hard to show that $\Phi$ increases by $\frac{q}{2} \ge \frac{1}{2}.$ If the filler can force this to happen for $n^3$ consecutive time steps, then $\Phi$ will increase to $\Omega(n^3)$, which means that at least one of the $|x_i|$'s must be $\Omega(n)$. If the filler is careful, then it turns out they can further ensure that the cups are symmetric (i.e., for every cup with fill $s$ there is another cup with fill $-s$). Thus if $|x_i| \ge \Omega(n)$ for some $n$, then a backlog of $\Omega(n)$ has been achieved. 

The only way that the filler might be prevented from performing this type of move for $n^3$ consecutive time steps is if, at some time step, no two cups have the same fills as each other. Note, however, that the filler always adds half-integer values to cups and the emptier always subtracts integer values, which means that $x_1, \dots, x_n$ are always half-integers. So, if the $x_i$'s are all distinct, then $\max_{1 \le i \le n} |x_i| \ge \frac{n-1}{4} = \Omega(n)$. Recalling that the filler can ensure symmetry of the cups, it follows that in this case the filler has also achieved an $\Omega(n)$ backlog in only $n^3$ time steps.

\paragraph{Considering smaller backlogs. }What if we only want to reach some backlog $o(n)$? We will now describe how the filler can achieve backlog $\Omega(t \log (n/t))$ in $O(t^3 \log (n/t))$ time steps. Combining this with some edge cases (which we defer to Section \ref{sec:lower}) results an optimal filling strategy for any backlog. 

First, we claim that within $t^3$ steps, the filler can cause $\Theta(n)$ cups to have fills $\Omega(t)$. To see this, note that if at least half of the cups have fills less than $t$, then by the pigeonhole principle, there must exist some half-integer $s$ with $|s| \le t$ and $\Omega(n/t)$ cups of fill exactly $s$. If the filler causes half of these cups to go up in fill by $1/2$ and half of them to go down in fill by $1/2$, then the net effect on $\Phi$ will be that it increases by $\Omega(n /t )$. The filler can repeatedly force $\Phi$ to increase by $\Omega(n / t)$, while keeping the maximum backlog at $t$, until either $t^3$ steps have passed and $\Phi = \Omega(n \cdot t^2)$ or until at least half of the cups have reached backlog $\pm t$. Either way, at least $\Theta(n)$ of the cups must have fills more than $\Omega(t)$ or less than $-\Omega(t)$, and recalling again that the filler can also ensure a symmetry of the cups, it follows that a constant fraction of the cups have fills $\Omega(t)$.

Once we have $cn$ cups of fill $\Omega(t)$, for some constant $c > 0$, the filler can focus on these cups, and force $c^2 n$ of these cups to fill $2 \cdot \Omega(t)$, then $c^3 \cdot n$ to fill $3 \cdot \Omega(t)$, and so on for a logarithmic number of phases. Overall, one can achieve backlog $\Omega(t \log (n/t))$ in $O(t^3 \log (n/t))$ time steps.\footnote{Note that one can only do $\log (n/t)$ phases, rather than the more intuitive $\log n$ phases, since we need at least $t$ cups to reach fill $a \cdot t$ in order for anything to reach fill $(a+1) \cdot t$.} 

\paragraph{The final piece: establishing that greedy emptying is optimal.}
To conclude our overview of the lowerbound, let us revisit the assumptions that (a) we are playing the negative-fill game, and (b) the emptier is playing greedily. The first assumption is trivial to remove, since the filler in the standard game is strictly better off than the filler in the negative-fill game. The second assumption, that the emptier is playing greedily, requires us to prove that the greedy emptying algorithm is always optimal (this ends up being useful to have for later results as well). 

To prove the greedy emptying is optimal, in Section \ref{sec:greedy_optimal}, we construct a specially designed poset on the possible states $X$ of the system. Say that a state $X = \{x_1, \ldots, x_n\}$ \defnn{weakly monopolizes} a state $Y = \{y_1, \ldots, y_n\}$ if it is possible to order the cups such that either:
\begin{itemize}
\item $x_i \ge y_i$ for all $i$; 
\item or we have (a) that $x_i = y_i$ for all $i > 2$, that (b) $x_1 > y_2$, and that (b) we can get from $X$ to $Y$ by removing exactly $ 1 $ unit of water from cup $ 1 $ and placing some quantity $0 \le c \le 1$ of water into cup $2$. 
\end{itemize}
The transitive closure of weak monopolization induces a partial ordering on the set of all possible system states, where $A \ge B$ if there is a sequence $A = A_1, A_2, \ldots, A_j = B$ such that each $A_i$ weakly monopolizes each $A_{i + 1}$. 

We prove that, given the choice between two states $A, B$ with $A \ge B$, the emptier should always prefer state $B$. Furthermore, starting from any state $X$, greedy emptying always results in a state $B$ that satisfies $B \le A$ for every other state $A$ that the emptier could have reached from $X$. Thus the greedy emptying algorithm is optimal.

\subsection{Overview of the Upper Bound on Backlog} \label{subsec:upper_overview}

In Section \ref{sec:upper}, we show that the greedy emptying algorithm achieves an upper bound on backlog matching the lower bound of Section \ref{sec:lower}.
We first consider a filler that is restricted to only making moves of the form described in Subsection \ref{subsec:lower_overview}, that is, so that the net effect of each round is that some number $2q$ of cups at some height $k$ are replaced with 
$q$ cups at height $k + 1/2$ and $q$ cups at height $k - 1/2$.
We will later show that these types of moves are (essentially) always optimal for the filler, which means this restriction is actually without loss of generality. 
Considering this restricted filler along with a greedy emptier leads to the following simple combinatorial problem, which we call the \defnn{stone game} (we call it the \defnn{stone-variant cup game} in Section \ref{sec:upper}):

\begin{quoting}
\noindent
Suppose you have $n$ stones on a number line, all starting at $0$. At each time step $t$, you may pick any point $k$ on the number line with $2$ or more stones, choose any integer $q \ge 1$ such that there are at least $2 \cdot q$ stones at $k$, and move $q$ of the stones at position $k$ to position $k-1$ and $q$ of the stones at position $k$ to position $k+1$. If you repeat this for $T$ time steps, what is the furthest that any of the $n$ stones may be from the origin?
\end{quoting}

\paragraph{Analyzing the time to get a stone to position $\Omega(n)$.} To analyze the stone game, we start by considering how long it takes for some stone to reach $b := n/10$ in absolute value. Our goal is to show that one needs at least $T = \Omega(n^3)$ time steps. To highlight the relationship between the cup game and the stone game, we shall sometimes refer to the distance of the furthest stone from the origin simply as the backlog. 

Let the positions of the stones be $x_1, \dots, x_n$. We start by recalling the potential function $\Phi = \sum_{i = 1}^{n} x_i^2$. One would naturally hope that $\Phi$ could help us establish a bound of $T= \Omega(n^3)$  (just as it helped us with the $T = O(n^3)$  bound in Subsection \ref{subsec:lower_overview}). For example, one way to prove the $\Omega(n^3)$ time bound would be to show that $\Phi$ increases by at most $O(1)$ in each step, and that $\Phi$ takes a value of at least $\Omega(n^3)$ after the final step (i.e., when backlog $b = n/10$ is achieved). Unfortunately, we run into two problems. The first problem is that even if $\max |x_i| = b = n/10,$ we may still have that $\sum_{i = 1}^{n} x_i^2$ is just $O(n^2)$ after the final step (as opposed to the desired $\Omega(n^3)$). The second, more difficult problem is that the number of stones we move in each direction could be large at each time step. If we move $q$ stones right and $q$ stones left, $\Phi$ increases by $2q$ (rather than the desired $O(1)$), which could be as large as $n$. Together, these two problems mean that, a priori, we can only get a trivial $T = \Omega(n)$ bound, since the change in the potential function $\Phi$ is up to $n$ at each time step and the final potential may be as small as $\Theta(n^2)$.

The first problem can be resolved with a more careful analysis: in particular, one can show that a backlog of $\Omega(n)$ actually does imply $\Phi = \Omega(n^3)$. The key is to prove that there can never be large gaps between consecutive stones. Namely, one can show that if there exists a stone at some position $k > 0$, there must be at least one stone at position $k-1$ or $k-2$, and likewise, if there exists a stone at $k < 0$, there must be at least one stone at position $k+1$ or $k+2.$ As a result, if there is a stone at $b = n/10,$ there must be a stone at position either $b-1$ or $b-2$, at either $b-3$ or $b-4$, and so on, so one can show that there are $\Omega(b)$ stones at positions $b/2$ or greater. Thus, if there is a stone at position $b$, then $\Phi \ge \Omega(b) \cdot (b/2)^2 \ge \Omega(n^3)$, as desired.

The more difficult piece of the analysis is to show that, even though $\Phi$ can grow significantly in a single step, $\Phi$ only increases \emph{on average} by $O(1)$ per step. An important insight here is to create a second potential function, $\Psi$, and compare the growth rates of $\Psi$ and $\Phi$. We define this new potential function $\Psi := \sum_{i < j} |x_i-x_j|,$ where $x_1, \dots, x_n$ are the locations of the $n$ stones. 
In a given step of the stone game, if we move $q$ stones up from $k$ to $k+1$ and $q$ stones down from $k$ to $k-1$, the first potential function $\Phi$ increases by $2q$. However, one can show that $\Psi$ must increase by at least $2q^2$ -- the core reason for this is that the $q$ stones that moved up now each have distance 2 from the $q$ stones that moved down, whereas before they had distance $0$. Now suppose for contradiction that $\Phi$ grows on average by some amount $\overline{q} = \omega(1)$ per step. This would mean that the average growth of $\Psi$ per step is at least $\Omega(\overline{q}^2)$, so the final values of $\Psi$ and $\Phi$ must satisfy $\Psi = \omega(\Phi)$.
However, we know that $\Phi = \Omega(n^3)$ at the end of the game, and it is not hard to see that $\Psi$ must be $O(n^3)$, since all of the $x_i$'s are bounded in the range $[-O(n), O(n)]$. Thus we have a contradiction, and the average growth of $\Phi$ per step is actually $O(1)$. 
The fact that $\Phi$ is $\Omega(n^3)$ at the end of the game and grows by $O(1)$ on average per step is sufficient to show that the time $T$ needed to achieve backlog $b = n / 10$ in the stone game is at least $\Omega(n^3)$ time steps. 

\paragraph{Considering smaller backlogs in the stone game. }But what happens if we wish to analyze the time needed to achieve backlog $k$ in the stone game, for $k \ll n$? If, at the end of the game, we knew that at least a constant fraction of the stones were in positions $\Omega(k)$, then we could use the same argument as the one above to show that the average growth rate of $\Phi$ is $O(1)$. Unfortunately, the number of stones at positions $\Omega(k)$ might be as small as $O(k)$. 
In this case, $\Phi$ could be as small as $O(k^3)$, whereas $\Psi$ must be at least $\Omega(k^2 n)$. If the average growth of $\Phi$ were $\Theta(n/k)$ and the average growth of $\Psi$ were $\Theta((n/k)^2)$, these values of $\Phi$ and $\Psi$ could be obtained in $O(k^4/n)$ rounds, which is much smaller than our desired bound of $\Omega(k^3 / \log^2(n / k))$.


To obtain the optimal bound of $\Omega(k^3/\log^2(n/k))$, we introduce another variant of the stone game that has what we call checkpoints: namely, we choose an integer $\ell$ and play the same game but now, once a stone has reached a position $a \cdot \ell$ for any $a \ge 0$, it can never go below it. In other words, if we move $q$ stones from $a \cdot \ell$ to $a \cdot \ell+1$, rather than moving $q$ stones from $a \cdot \ell$ down to $a \cdot \ell - 1$, we keep them as is. A key insight is that, if a player wishes to get a stone to position $10 \log n \cdot \ell$, the player must have the property that, for the majority of the checkpoints, the player gets at least half of the stones that reach checkpoint $a \cdot \ell$ all the way up to checkpoint $(a+1) \cdot \ell$. 

We can think of the steps of the stone game (with checkpoints) as being split into subgames, where each subgame takes place between two checkpoints $a \cdot \ell$ and $(a + 1) \cdot \ell$. We analyze each subgame individually by creating potential functions $\Phi_a, \Psi_a$ between each pair of checkpoints at $a \cdot \ell$ and $(a+1) \cdot \ell$. At least half of the subgames have the property that at least half of their stones make it to the next checkpoint, and this property makes each such subgame amenable to being analyzed using $\Phi_a$ and $\Psi_a$. By analyzing the subgames individually, we can show that the total length of the full stone game (which is the sum of the lengths of the subgames) is at least $\Omega(k^3/\log^2 (n/k))$. 

Finally, one can show that adding the checkpoints to the original stone game only makes it easier for the player to achieve a large backlog. 
This involves proving that if states $X$ and $Y$ have $x_i \ge y_i$ for all $i$, then given any stone-game operation on $Y$ to obtain a state $Y'$, we can generate a corresponding (but perhaps different) stone-game operation on $X$ to obtain $X'$, and preserve the ordering $x_i' \ge y_i'$ for all $i$.
Thus, the $T = \Omega(k^3/\log^2 (n/k))$ bound also applies to the original stone game.

\paragraph{Transferring bounds from the stone game to the cup game.} Overall, the potential-function arguments above get optimal bounds for the \emph{stone game} but they do not directly give bounds for the more general variable-processor cup game. The main issue is that in certain time steps in the cup game, it might be possible for the filler to make ``backward'' moves where both $\Phi$ and $\Psi$ decrease considerably: in the worst case, $\Psi$ can even decrease by up to $n^2$. Our analysis of the stone game relied heavily on the fact that the change in $\Psi$ was comparable to the square of the change in $\Phi$ (in a given step), but if $\Psi$ can decrease significantly in a single step, then our comparison of growth rates no longer works. Our fix for this is not to modify the potential functions, but rather to show that these backward moves, or any other ``non-stone game'' moves, are never advantageous to the filler, even in the long run.

To show that non-stone-game moves are never advantageous for the filler, we analyze the relationship between the stone game and the variable-processor cup game. 
Say that a sequence $X = \{x_1, \dots, x_n\}$ of real numbers \defnn{majorizes} another sequence $Y = \{y_1, \dots, y_n\}$ (where $x_1 \ge \cdots \ge x_n$ and $y_1 \ge \cdots \ge y_n$) if $x_1+ \cdots+x_i \ge y_1+\cdots+y_i$ for all $1 \le i \le n$, and if $\sum_i x_i = \sum_i y_i$. We prove that, for any sequence of cup game rounds, one can create a corresponding sequence of stone game rounds such that after every round, the sequence of stone positions majorizes the sequence of cup fills. Since $X$ majorizing $Y$ implies that $x_1 \ge y_1$, we get that the maximum backlog after $T$ steps of the cup game is at most the maximum backlog after $T$ steps of the stone game. Therefore, even if we cannot utilize the potential functions on the general cup game, it suffices to look at the stone game as it will have a greater maximum backlog.

The main technical claim needed to show this majorization result is to show that if a sequence $X$ majorizes a sequence $Y$, and if $Y$ can be converted to a sequence $Y'$ in one round of filling/greedy-emptying in the variable-processor cup game, then it is possible to convert $X$ into some $X'$ (also with a single round of filling/greedy emptying) such that $X'$ majorizes $Y'$. This claim is quite casework-heavy and crucially uses the fact that we are in the \emph{variable-processor} cup game. Indeed, the choice of $p$ may have to differ between the round performed on $X$ and the round performed on $Y$, and perhaps surprisingly, the claim is actually \emph{false} in the fixed $p$-processor cup game.

\subsection{Overview of Resource-Augmentation Analysis} \label{subsec:resource_augmentation_overview}

In Section \ref{sec:resource_augmentation}, we analyze the variable-processor cup game with $\eps $ resource augmentation (and non-negative fills), meaning that in each time step, the emptier is permitted to remove up to $ 1+\eps $ units of water from each of $ p $ cups (rather than just $ 1 $ unit of water).

We prove that even a very small amount of resource augmentation significantly decreases backlog of the game: the greedy emptying algorithm achieves backlog $ O (\eps ^ { -1 }\log n) $.

The proof uses the probabilistic method. Rather than analyzing the greedy emptying algorithm directly, we instead construct a randomized emptying algorithm that, at any given moment, achieves backlog $ O (\eps ^ { -1 }\log n) $ with \emph{non-zero} probability. (Importantly, the randomized algorithm is against an \emph{adaptive} filler, not an oblivious one.) The existence of such a randomized algorithm non-constructively implies the existence of a deterministic emptying algorithm with the same guarantee. But we already know that the best deterministic emptying algorithm is the greedy one. Thus greedy emptying achieves backlog $ O (\eps ^ { -1 }\log n) $ (deterministically). 

To simplify our discussion here, let us consider only games of polynomial length (in Section \ref{sec:resource_augmentation} we consider arbitrary game lengths). In this case, our randomized emptying algorithm can simply take an approach that we call  \defnn{proportional emptying}: in each time step of the game, if the filler places some amount $ q_j $ of water into cup $ j $, then the emptier empties from cup $ j $ with probability exactly $ q_j $. 

To analyze proportional emptying, we show that, at any given moment, each cup has fill $ O (\eps ^ { -1 }\log n) $. Roughly speaking, the amount of water in each cup can then be modeled as a biased random walk: in each time step, the expected amount of water that the emptier removes is a factor of $ 1+\eps $ larger than the amount of water that the filler inserts. The bias in the random walk prevents it from ever reaching a large fill. The result is that a simple Chernoff-style analysis (modified using a variation on Azuma's martingale inequality to handle the fact that the filler is an adaptive adversary) can be used to bound the fill in each cup by $ O (\eps ^ { -1 }\log n) $ (with high probability).

Interestingly, the above argument also immediately yields a nontrivial bound in the resource-augmentation-free setting. Now the amount of water in each cup follows an \emph{unbiased} random walk. At any given time step $t$, one can bound the height of such random walk by $O(\sqrt{t \log n})$ with high probability. Using the fact that greedy emptying is as good as any randomized emptying strategy, it follows that greedy emptying achieves backlog $O(\sqrt{t \log n})$ in a $t$ time step game.

\section{The Greedy Emptier is Always Optimal} \label{sec:greedy_optimal}

Intuitively, the greedy emptying algorithm (i.e., always empty from the fullest cups) should be the optimal deterministic emptying algorithm (for both the $ p $-processor cup game and the variable-processor cup game). This intuition is known to be true for the single-processor cup game starting in a state with all empty cups (in particular, the lower and upper bounds on backlog match in this case \cite{AdlerBeFr03}), but whether or not the intuition is correct in general has remained an open question. (We remark that there are variants of the game, for example the fixed-rate cup game, where greedy emptying is \emph{not} optimal, even asymptotically \cite{AdlerBeFr03}.) In this section, we prove that greedy emptying is, in fact, optimal for both of the $ p $-processor cup game and the variable-processor cup game.. That is, for any starting state of the cups, and for any game length, greedy emptying is the best possible algorithm for minimizing backlog against an adaptive filler.

For any state $S$ and any length $t$, define $\OPT(S, t)$ to be the supremum backlog that a filler can achieve at the end of a $t$-step game starting at state $S$ assuming the emptier plays optimally. That is, $\OPT(S, 0)$ is just the amount of water in the fullest cup of $S$, and $\OPT(S, t)$ for $t > 0$ is defined by induction as
$$\sup_{S' \text{ reachable from }S \text{ by filler}} \phantom{foo} \min_{S'' \text{ reachable from }S'\text{ by emptier}} \OPT(S'', t - 1).$$
Note that this also allows for us to talk about the \defnn{optimal emptier}, which is the emptier that achieves backlog at most $\OPT(S, t)$ in any $t$-step game starting at any state $S$.

Define $\GREEDY(S, t)$ to be the supremum backlog that a filler can achieve at the end of a $t$-step game starting at state $S$, assuming the emptier plays greedily.
We wish to prove that $\OPT(S, t) = \GREEDY(S, t)$ for all $S, t$. Throughout the section we shall prove all of our results for both the versions of the games with non-negative fills and the versions of the games with negative fills.

We say that a state $ A $ \defnn{monopolizes}  a state $ B $ if it is possible to assign labels $ 1, 2,\ldots, n$ to the cups in $ A $ and $ B $ such that:
\begin{itemize}
\item cups $ 3, 4,\ldots, n $ contain the same amounts of water in both states;
\item cup 1 in $ A $ contains more water than cup 2 in $ B $;
\item cup 1 in $A$ contains one more unit of water than cup 1 in $B$;
\item cup 2 in $B$ contains $c$ more units of water than cup 2 in $A$ for some $0 \le c \le 1$.
\end{itemize}
In other words, you can get from $ A $ to $ B $ by removing 1 unit of water from cup 1 and and placing $c \le 1$ units into cup 2, and cup 1 in $A$ contains more water than cup 2 in $B$. 

We say that $ A $  \defnn{dominates} $B$ if it is possible to label the cups such that each cup $ i $ in $A$ contains at least as much water as cup $i$ in $B$. Finally, we say that $ A $ \defnn{weakly monopolizes} $ B $ if either $ A $ monopolizes $ B $, or $ A $ dominates $ B $. 

We now prove several properties of weak monopolization in the $p$-processor cup game (for any $p$). Our first lemma says that, if $ A $ weakly monopolizes $ B $, then for any filler move on $B$, there is some filler move on $A$ that preserves the weak monopolization of $A$ over $B$.
\begin{lemma} \label{lem:dom1}
Suppose $ A $ weakly monopolizes $ B $. Suppose that there is a $p$-processor filler move that transforms $B$ into some $B'$. Then there exists a $ p $-processor filler move that transforms $A$ into some $A'$ such that $A'$ weakly monopolizes $ B' $.  
\end{lemma}
\begin{proof}
If $ A $ dominates $ B $, then the lemma is trivial. Thus we can assume that $A$ monopolizes $B$. Let $q$ be the amount by which cup 1 in $A$ contains more water than cup 2 in $B$.

Define $ X $ to be a state that we reach from $ A $ if we perform the same filler move that transforms $B$ into $ B' $. If cup 1 in $X$ contains more water than cup 2 in $B'$ then we can set $A' = X$ and be done.

Suppose cup 1 in $X$ does not contain more water than cup 2 in $B'$. Then, in the transformation from $B$ to $B'$, the filler must have placed at least $q$ more water into cup 2 than into cup 1. Let $r$ be the amount of water that the filler placed into cup $1$, and let $r + q + s$ (for some $s \ge 0$) the amount of water that the filler placed into cup $2$. 

Now define $A'$ to be the state that one would reach from $A$ by performing the following $p$-processor filler move: place $r + s$ units of water into cup 1 and $q + r$ units of water into cup 2 (and then place water into cups $3, 4, \ldots, n$ in the same way as to transform $B$ to $B'$). Cup 1 in $ A' $ contains the same amount of water as cup 2 in $ B' $, and cup 2 and $ A' $ contains the same amount of water as cup 1 in $ B' $. Thus $A'$ and $B'$ are equivalent states, meaning that $A'$ weakly monopolizes $B'$.
\end{proof}

Our next lemma says that, if $ A $ weakly monopolizes $ B $, and if $ A $ is then greedily emptied from, then it is possible to empty from $ B $ in such a way that the monopolization relationship is preserved.
\begin{lemma} \label{lem:dom2}
Suppose $A$ weakly monopolizes $B$. Let $A'$ be the state reached if a $p$-processor emptier greedily empties from $A$. Then there exists a valid $ p $-processor emptier move on $ B $ that achieves some state $B'$ such that $A'$ weakly monopolizes $B'$. 
\end{lemma}
\begin{proof}
If $ A $ dominates $ B $, then the lemma is trivial. Thus we can assume that $A$ monopolizes $B$.

Let $a_1, a_2$ denote the fills of cups $1$ and $2$ in $A$, and let $b_1, b_2$ denote the fills of cups $ 1 $ and $ 2 $ in $B$. So $a_1 = b_1 + 1$, $b_2 = a_2 + c$ for some $0 \le c \le 1$, and $a_1 > b_2$. Furthermore, let $a_1', a_2'$ the fills of cups $1$ and $2$ in $A'$, and let $b_1', b_2'$ denote the fills of cups $ 1 $ and $ 2 $ in $B'$ (once we've defined it).

If the greedy emptier on $ A $ empties from neither cup 1 nor cup 2, then the same set of empties on $B$ results in a $B'$ that is weakly monopolized by $A'$.

Next consider the case where the greedy emptier on $A$ empties from both cups 1 and 2. Then we empty from the same set of cups in $B$ to arrive at a state $B'$. In the negative-fill game, is immediate to see that $A'$ monopolizes $B'$. In the standard game, we must be careful about the fact that one of $b_1$ or $a_2$ might be less than $1$. If $ b_1 < 1$, then $b_1' = 0 \le a_2'$, and since $a_1 > b_2$, we also have $a_1' \ge b_2'$; thus $A'$ dominates $B'$ in this case. If $a_2 < 1$ but $b_1 \ge 1$, then we have $a_1' = b_1' + 1$ and $b_2' = a_2' + c'$ for some $0 \le  c \le 1$; we also have $a_1' \ge b_2'$ (since $a_1 \ge b_2$), so $A'$ monopolizes $B'$ in this case.

Finally, consider the case where the greedy emptier on $ A $ empties from only one of the cups 1 or 2. Since $a_1 > b_2 \ge a_2$, the emptier must empty from cup 1. Suppose we construct $B'$ by performing the same set of empties on $B$ as were performed on $A$, but we remove water from cup 2 instead of cup 1. Then $a_1' = b_1'$ and $a_2' \ge b_2'$, so $A'$ dominates $B'$. This completes the proof.
\end{proof}

By combining the previous two lemmas, we can arrive at the following.
\begin{lemma} \label{lem:dom3}
Consider either the $p$-processor cup game, for some $ p $, or the variable processor cup game (and consider either the negative-fill case or the non-negative-fill case).
Suppose $A$ weakly monopolizes $B$. Then for any $t$, 
$$\GREEDY(A, t) \ge \OPT(B, t).$$
\end{lemma}
\begin{proof}
We prove the lemma by induction on $ t $. In the base case of $ t = 0 $, the lemma reduces to showing that $ A $ has backlog at least as large as $ B $; this follows from the fact that $ A $ weakly monopolizes $ B $.

Now consider some $ t > 0 $, and suppose the result holds for $t - 1$. For each state $B'$ that the filler can reach from $B$, Lemma \ref{lem:dom1} tells us that there must exist a state $A'$ that the filler can reach from $A$ in such that $ A' $ weakly monopolizes $ B' $. Moreover, Lemma \ref{lem:dom2} tells us that, if greedy emptying from $ A' $ results in some state $A''$, then there must be an emptying move on $B'$ that results in a state $B''$ such that $A''$ weakly monopolizes $ B'' $.

Note that, in the previous paragraph, the filler's move from $B$ to $B'$ (i.e., the choice of $B'$) fully determines $A'$ (by the construction in Lemma \ref{lem:dom1}), $A''$ (by greedy emptying on $A'$), and $B''$ (by the construction in Lemma \ref{lem:dom2}).
Thus, we will think of $A', A'', B''$ as being functions of $B'$.

Expanding out $\GREEDY(A, t)$, we have
$$\GREEDY(A, t) \ge \sup_{B'} \GREEDY(A'', t - 1).$$
By the inductive hypothesis, since each $A''$ weakly monopolizes $B''$, we have $\GREEDY(A'', t - 1) \ge \OPT(B'', t - 1)$. Thus
$$\GREEDY(A, t) \ge \sup_{B'} \OPT(B'', t - 1).$$
Since $\sup_{B'} \OPT(B'', t - 1) \ge \OPT(B, t)$, the lemma follows.
\end{proof}

Recall that our goal is to upperbound $\GREEDY(A, t)$ by $\OPT(A, t)$ for all $A, t$. Thus Lemma \ref{lem:dom3} may at first seem to be backward progress, since the lemma establishes a \emph{lower bound} on $\GREEDY(A, t)$. The trick to using Lemma  \ref{lem:dom3} is that we will only ever apply the lemma to states $A$ and game lengths $t$ for which we have already proven by induction that $\OPT(A, t) = \GREEDY(A, t)$; in this setting, the lemma establishes that $\OPT(A, t) \ge \OPT(B, t)$ which, as we shall see, ends up being critical to the proof.

\begin{theorem} \label{thm:greedyopt}
Consider either the $p$-processor cup game, for some $ p $, or the variable-processor cup game (and consider either the negative-fill case or the non-negative-fill case).
For all states $ A $ and game lengths $t$, 
$$\GREEDY(A, t) = \OPT(A, t).$$
\end{theorem}
\begin{proof}
We prove the theorem by induction on $ t $ with a trivial base case of $ t = 0$. Consider $ t >0 $, and suppose the theorem holds for $t-1 $. 

Consider any state $ A' $ that the filler can reach from $ A $. Let $ A'' $ be the state reached from $A'$ if the emptier empties greedily, and let $B$ be the state reached from $A'$ if the emptier empties according to OPT. Since the transformation from $A'$ to $A''$ empties from cups with at least as much water as the transformation from $A'$ to $B$, there must be a sequence of states $X_0, X_1, \ldots, X_k$, with $X_0 = B$ and $X_k = A''$, such that each $X_i$ weakly monopolizes $X_{i - 1}$. In particular, one can define the $X_i$'s such the only difference between each $X_i$ and $X_{i + 1}$ is that, to get from $A'$ to $X_{i + 1}$ instead of from $A'$ to $X_i$, the emptier removes water from a cup $j$ instead of a cup $k$, where cup $j$ contains less water than cup $k$ in state $A'$. It is straightforward to verify that this results in each $X_{i + 1}$ weakly monopolizing each $X_i$: if cup $k$ (in state $A'$) has fill at least $1$, then $X_{i + 1}$ monopolizes $X_i$; and otherwise, both cups $k$ and $j$ (in state $A'$) have fills less than $1$, and thus $X_{i + 1}$ dominates $X_i$.

By the inductive hypothesis, we know that $\OPT(X_i, t - 1) = \GREEDY(X_i, t - 1)$ for each $X_i$. Thus Lemma \ref{lem:dom3} tells us that $\OPT(X_i, t - 1) \ge \OPT(X_{i - 1}, t - 1)$. By transitivity, it follows that $\OPT(X_k, t - 1) \ge \OPT(X_0, t - 1)$, and thus $\OPT(B, t - 1) \ge \OPT(A'', t - 1)$. Again applying the inductive hypothesis to deduce that $\OPT(A'', t - 1) = \GREEDY(A'', t - 1)$, we have that 
$$\OPT(B, t - 1) \ge \GREEDY(A'', t - 1).$$
Thus
$$\GREEDY(A, t) = \sup_{A''} \GREEDY(A'', t - 1) \le \sup_{B} \OPT(B, t - 1) = \OPT(A, t).$$
Finally, as $\GREEDY(A, t) \ge \OPT(A, t)$ trivially, we have $\GREEDY(A, t) = \OPT(A, t)$.
\end{proof}


\begin{corollary}
Theorem \ref{thm:greedyopt} continues to hold for the game with non-negative fills even if the emptier is given $1 + \eps$ resource augmentation.
\label{cor:aug}
\end{corollary}
\begin{proof}
The proof of Theorem \ref{thm:greedyopt} continues to hold without modification. The only difference is that, now, we say that a state $A$ monopolizes a state $B$ if there is a ordering of the cups for which two properties hold: (1) we can take $1 + \eps$ unit of water from some cup $1$ in $A$, place some amount $0 \le c \le 1 + \eps$ of water into some other cup $2$ in $A$, and in doing so arrive at $B$; and (2) cup $1$ in $A$ contains more water than cup $2$ in $B$. 
\end{proof}

\section{Lower Bounding Backlog} \label{sec:lower}

In this section, we prove the optimal lower bound on the backlog. 
In other words, we show that for any integer $t \ge 1$, there exists a filling strategy that can guarantee a backlog of $\Omega(b(t))$ against \emph{any} emptying strategy, after $t$ rounds of the variable-processor cup game, where $b(t)$ is defined as in Equation \eqref{trade-off-curve} in the introduction (Section \ref{sec:intro}). This result can be captured compactly in the following theorem:

\begin{restatable}{theorem}{thmlower}
\label{thm:main_lower}
    For some absolute constant $c > 0$ and any $k \le c \cdot n$, the filler in the variable-processor cup game can force a maximum backlog of $\Omega(k)$ using only $O\left(k + \frac{k^3}{\log^2 (n/k)}\right)$ time steps.
\end{restatable}

We remark that due to the optimality of greedy emptying, which we proved in the previous section, we focus in this section on designing filling strategies that are effective against the greedy emptier.

In Subsection \ref{subsec:lower_warmup}, we give a simple proof for obtaining $\Omega(n)$ backlog in $O(n^3)$ rounds. In Subsection \ref{subsec:lower}, we generalize Subsection \ref{subsec:lower_warmup} and develop the full proof of Theorem \ref{thm:main_lower}.

\subsection{Warmup: Achieving a $\Theta(n)$ Backlog Quickly Against a Greedy Emptier} \label{subsec:lower_warmup}

In this subsection, we show that in only $n^3$ rounds, a filler can achieve a backlog of $\frac{n-1}{2}$ against a greedy emptier.
While the following argument is simpler than the more refined lower bound in Subsection \ref{subsec:lower}, it conveys much of the intuition behind the general lower bound.
Here, we only consider the standard variable-processor cup game, though in subsection \ref{subsec:lower} we will consider both the standard and negative-fill games.

\begin{theorem}
    Against a greedy emptier in the standard variable-processor cup game, the filler can achieve a backlog of $\frac{n-1}{2}$ in only $n^3$ rounds.
\end{theorem}

\begin{proof}
    The filler follows the following simple strategy, which will guarantee that the fills at the end of each round are half integers (i.e., integer multiples of $\frac{1}{2}$). Suppose the cups after some round, in sorted order, are $x_1 \ge x_2 \ge \cdots \ge x_n$, which are all half-integers. 
    
    If $x_i > x_{i+1}$ for all $1 \le i \le n-1,$ then we must have that $x_i \ge x_{i+1} + \frac{1}{2}$ for all $i \le n-1,$ so $x_1 \ge \frac{n-1}{2}$. In this case, we have already achieved a backlog of $\frac{n-1}{2}.$
    
    Otherwise, we consider the potential function $\Phi = \sum_{i = 1}^{n} x_i^2,$ i.e., the sum of the squares of the backlogs. Consider the smallest $p$ such that $x_p = x_{p+1}$. We set the number of processors for the round to be $p$, we place $1$ unit of water into each of cups $x_1, \dots, x_{p-1}$ and we place $1/2$ unit of water into each of cups $x_p, x_{p+1}$. Then, the greedy emptier will empty $1$ unit from each of the first $p-1$ cups and will choose to empty $1$ unit from either the $p^{\text{th}}$ cup or the $(p+1)^{\text{th}}$ cup (WLOG, they choose the $(p + 1)$-th cup). This means that, over the course of the entire round, the first $p-1$ cups are unchanged, the $p^{\text{th}}$ cup goes up by $\frac{1}{2}$, and $(p+1)^{\text{th}}$ cup goes down by $\frac{1}{2}$. The only exception is if $x_p = x_{p+1} = 0$ at the beginning, in which case one of these two cups will go up to $1/2$ and the other remains at $0$. Clearly, all of the fills remain half-integers at the end of each round. If we replace $x_p$ and $x_{p+1}=x_p$ with $x_p+\frac{1}{2}$ and $x_p-\frac{1}{2}$, then the sum of the squares of the backlogs increases by $1/2$. Otherwise, if we replace $x_p = x_{p+1} = 0$ with $0$ and $1/2$ in some order, then the sum of the squares of the backlogs increases by $1/4.$ Thus, at the end of each round, unless one of the backlogs was already $\frac{n-1}{2}$ or greater, $\Phi$ increases by at least $\frac{1}{4}.$
    
    Therefore, after $n^3$ rounds, either at some point we will have achieved $x_1 > x_2 > \cdots > x_n$ and thus a backlog of $\frac{n-1}{2}$, or we have that $\sum_{i = 1}^{n} x_i^2 \ge \frac{n^3}{4}$, which would mean that $\max x_i \ge \frac{n}{2}$, since all of the $x_i$'s are nonnegative. As a result, there exists a filling strategy that can guarantee a $\frac{n-1}{2}$ backlog against a greedy emptier in at most $n^3$ rounds.
\end{proof}

\subsection{The General Lower Bound}\label{subsec:lower}

In this subsection, we prove the lower bound on the backlog for the variable-processor cup game after $t$ rounds, showing that the optimal backlog is $\Omega(b(t))$ for $b(t)$ defined as in Equation \eqref{trade-off-curve}. Equivalently, we show that for any integer $k \le n$, a filler can achieve backlog $\Omega(k)$ in $O\left(k + \frac{k^3}{\log^2 (n/k)}\right)$ rounds against any emptying strategy. We remark that if we are playing $t$ rounds of the game and the filler has achieved backlog $b$ by round $t' < t$, the filler can ensure the backlog stays at $b$ by setting $p = n$ and filling every cup for steps $t'+1, \dots, t.$ Thus, we just need to show the filler can obtain this backlog \emph{within} this many rounds.

Throughout this subsection, we assume that we are playing the negative-fill game, and that the emptier is always greedy. We remove both of these assumptions at the end in Theorem \ref{thm:main_lower}.

\begin{lemma} \label{lem:basic_movement}
    Assume that the emptier is always greedy, and that at the beginning of some round, the fills are all (possibly negative) half-integers. For some integer $k \in \BZ$ and positive integer $q \in \BN,$ suppose there are at least $2q$ cups having fill exactly $k/2$. Then, the filler can ensure that at the end of the round (i.e., after both the filler and emptier move), exactly $q$ of the cups of fill exactly $k/2$ have increased to $(k+1)/2,$ exactly $q$ of the cups of fill exactly $k/2$ have decreased to $(k-1)/2,$ and all remaining cups are unchanged.
\end{lemma}

\begin{proof}
    Suppose the cups are sorted in order $x_1 \ge x_2 \ge \cdots \ge x_n$, such that $x_i = x_{i+1} = \dots = x_{i+2q-1} = k/2$. Also, suppose $i$ is the smallest integer such that $x_i = k/2,$ so either $x_{i-1} > k/2$ or $i = 1$. Then, the filler will set $p = i-1+q$ and fill the first $i-1$ cups with $1$ unit of water and the cups $x_i, \dots, x_{i+2q-1}$ each with $1/2$ unit of water. Then, if $i > 1,$ the emptier is forced to empty $1$ unit of water from each of the first $i-1$ cups. In addition, we have that the cups $i, i+1, \dots, i+2q-1$ all have fills $(k+1)/2$, which is at least $1/2$ unit of water more than all later cups, which means that the emptier will remove $1$ unit of water from exactly $q$ of the cups $i, i+1,\dots, i+2q-1.$ Therefore, all cups $1, \dots, i-1$ are unchanged (since $1$ unit of water is added and then removed) and all cups $i+2q, \dots, n$ are unchanged (since water is never added nor removed). Finally, among the cups $i, i+1, \dots, i+2q-1$, exactly $q$ of them will end up at $(k-1)/2$ and exactly $q$ of them will end up at $(k + 1)/ 2$. 
\end{proof}

\begin{lemma} \label{lem:advanced_movement}
    Suppose that $k, m$ are positive integers such that $4 k|m$, and suppose there are at least $m$ cups that currently have fill $0$ (where $n \ge m$ is the total number of cups). Then, against a greedy emptier, the filler can ensure that at least $m/4$ cups will have fill exactly $k/2$ after $O(k^3)$ rounds.
\end{lemma}

\begin{proof}
    WLOG assume that the first $m$ cups have fills $x_1, x_2, \dots, x_m$, and at the start, $x_1 = x_2 = \cdots = x_m = 0$. (We do not assume the cups are sorted in order of fill.) We only ever modify the first $m$ cups, and after each step, we will maintain the following invariants:
\begin{enumerate}
    \item All of the fills are half-integers.
    \item For each integer $j$, the number of cups of fill $j/2$ equals the number of cups of fill $-j/2$. 
    \item For each integer $j$, the number of cups of fill $j/2$ is always a multiple of $\frac{m}{4k}$.
    \item No cup has fill greater than $k/2$ or less than $-k/2$.
\end{enumerate}
    
    Trivially, all $4$ of these invariants are true at the beginning, since all the fills are $0$ and $\frac{m}{4k}|m$.
    Our procedure is the following. Suppose that there exists some integer $j$ such that $-k < j < k$ and the number of cups of fill $j/2$ is at least $\frac{m}{2k}$. Then, we use Lemma \ref{lem:basic_movement} to move $\frac{m}{4k}$ of these cups to fill $(j+1)/2$ and move $\frac{m}{4k}$ of these cups to fill $(j-1)/2$. In addition, if $j \neq 0$, then we know the number of cups of fill $-j/2$ is at least $\frac{m}{2k}$. So, in the next step, we move $\frac{m}{4k}$ of these cups to fill $(-j+1)/2$ and move $\frac{m}{4k}$ of these cups to fill $(-j-1)/2$. When $j \neq 0$, we perform these two steps consecutively as a pair. We keep repeating these types of steps (and pairs of steps) until we can no longer do so. It is clear that the first three invariants are preserved, and the last is preserved since we only modify cups that have fill between $-(k-1)/2$ and $(k-1)/2$, and we change their fills by at most $1/2$ per round.
    
    Now, we note that at each step, the potential function $\Phi = \sum_{i = 1}^{m} x_i^2$ increases by $\frac{m}{8k}$, since $$\frac{m}{4k} \cdot \left(\frac{j+1}{2}\right)^2 + \frac{m}{4k} \cdot \left(\frac{j-1}{2}\right)^2 - \frac{m}{2k} \cdot \left(\frac{j}{2}\right)^2 = \frac{m}{8k}.$$ This also means that if we do a pair of steps (i.e., modifying cups of fill $j/2$ and $-j/2$), the potential function $\Phi$ increases by $\frac{m}{4k}.$
    
    Now, after $4k^3$ steps (and pairs of steps), $\Phi$ is at least $4k^3 \cdot \frac{m}{8k} = \frac{1}{2} \cdot m k^2$. But this is impossible, since all of the fills of cups $1$ through $m$ are between $k/2$ and $-k/2$, so the maximum possible value of $\Phi$ is $\frac{1}{4} \cdot m k^2.$ This means that before reaching $4k^3$ of steps or pairs of steps, we must not be able to do any more steps. That is, for all $j$ such that $-k < j < k,$ there are less than $\frac{m}{2k}$ cups of fill exactly $j/2$, and because of our third invariant, this means there are at most $\frac{m}{4k}$ cups of fill exactly $j/2$. Therefore, the number of cups of fill between $-k+1$ and $k-1$ is at most $(2k-1) \cdot \frac{m}{4k} \le \frac{m}{2}.$ So, at least $\frac{m}{2}$ cups must have fill $\pm \frac{k}{2},$ and by the second invariant, this means at least $\frac{m}{4}$ cups must have fill $\frac{k}{2}.$ Reaching this stage required at most $4k^3$ steps and pairs of steps, which means the filler can succeed in at most $8k^3$ steps.
\end{proof}

\begin{corollary} \label{cor:force_upward}
    Suppose $k, m$ are positive integers such that $4k|m$, and suppose there are at least $m$ cups that currently have fill exactly $t$ for some real number $t$. Then, after $O(k^3)$ time steps, the filler can force at least $m/4$ of these cups to have fill exactly $t+\frac{k}{2}$, against a greedy emptier.
\end{corollary}

\begin{proof}
    We can apply same argument as in Lemma \ref{lem:advanced_movement}, where all of the fills are shifted up by the same number $t$. So, after at most $8k^3$ time steps, at least $m/4$ of the cups will have fill exactly $t + \frac{k}{2}$.
\end{proof}

From here, we are able to conclude with our main result of this section.

\thmlower*
\begin{proof}
    For now we will assume that the emptier is greedy, and that we are playing the negative-fill cup game. We will remove both assumptions at the end of the proof. 
    
    First, suppose $k \ge c \log n$. In this case, let $k' := \left\lceil \frac{k}{c \log (n/k)} \right\rceil \ge 2$. Also, let $n'$ be such that $n \ge n' > \frac{n}{4}$ and $n' = k' \cdot 4^r$ for some integer $r \ge 1$.
    
    Since there are at least $n' = k' \cdot 4^r$ cups of fill $0$ at the beginning of the game, we can apply Corollary \ref{cor:force_upward} with $t = 0$ and $m = n'$ to make sure there are at least $k' \cdot 4^{r-1}$ cups of fill $k'$ after $O(k'^3)$ time steps, since $4k'|n'.$ Then, we can again apply Corollary \ref{cor:force_upward} with $t = k'$ and $m = n'/4$ to make sure there are at least $k' \cdot 4^{r-2}$ cups of fill $2k'$ after an additional $O(k'^3)$ time steps. We can repeat this $r$ times until we have at least $k'$ cups of fill $r \cdot k'$ after a total of $O(r \cdot k'^3)$ time steps. But since $r = \log_4 (n'/k') = \Theta(\log (n/k \cdot \log (n/k))) = \Theta(\log (n/k)),$ this means that the filler can achieve a maximum backlog of $r \cdot k' = \Theta(k)$ using only $O(r \cdot k'^3) = O(k^3/\log^2 (n/k))$ time steps.
    
    Next, suppose $k < c \log n$. In this case, set $n' = 2^k$: if $c < 1$ then $n' < n$. By Lemma \ref{lem:basic_movement}, we can send $2^{k-1}$ cups to fill $1$ during the first round, then $2^{k-2}$ of those cups to fill $2$ during the second round, and so on until we have $1$ cup at fill $k$ after $k$ rounds.
    
    Up until now we have assumed that the filler was operating against a greedy emptying algorithm in the negative-fill cup game. However, by Theorem \ref{thm:greedyopt}, if the emptier uses a different emptying algorithm, the filler could always modify their strategy to get equal or greater backlog in the same amount of time. So, regardless of the emptying strategy, the filler can force a maximum backlog of $\Theta(k)$ using $O\left(k + \frac{k^3}{\log^2 (n/k)}\right)$ time steps, assuming we are playing the negative-fill cup game. 
    Finally, recall that the maximum backlog that the filler can ensure in the negative-fill game is at most the maximum backlog that the filler can ensure in the standard game for any fixed number of time steps.
    Thus, if  the filler can force a maximum backlog of $\Theta(k)$ using $O\left(k + \frac{k^3}{\log^2 (n/k)}\right)$ time steps in the negative-fill game, then
    the filler can force the same backlog using the same number of time steps in the standard game, which completes the theorem.
\end{proof}

Finally, to conclude the section, we use the optimality of greedy emptying to resolve an open question of \cite{variable} concerning whether a filler
who is limited in the speed at which they can change $p$ can still force a large backlog.
\begin{proposition}
Consider the variable-processor cup game, starting with all empty cups (and with negative fills either allowed for disallowed).
Suppose the filler is restricted to only change $p$ once every $n^c$ steps, for some constant $c \ge 0$, and to only change $p$ by $\pm 1$ at a time. 
Then the filler can still force a backlog of $\Omega(n)$ in polynomial time.
\label{prop:limitedfiller}
\end{proposition}
\begin{proof}
Call this version of the game the \emph{change-limited variable-processor cup game}. Notice that the proof of optimality for greedy emptying never changes the value of $ p $ used in any given step. Thus the same proof implies that greedy emptying is optimal in the change-limited variable-processor cup game. Henceforth we will assume that the emptier is greedy.

Next, observe that, no matter the value of $ p $, the filler can always ``skip'' a step in the following way: the filler simply places $ 1 $ unit of water into each of the $ p $ fullest cups, forcing the greedy emptier to remove water from those cups, so that the step has no net effect.

We have already proven that, in the standard variable-processor cup game, the filler can cause backlog $\Omega (n) $ in polynomial time. The filler can use the same strategy here, but with the following modification: between any two consecutive steps that use $p$ and $q$ processors, respectively, the filler spends $|p - q| n^c$ steps changing the number of processors from $p$ to $q$. That is, for each $i \in [p, q]$, the filler spends $n^c$ steps with $i$ processors, and renders each of those steps to have no net effect. The result is that, in polynomial time, the filler can force backlog $\Omega(n)$, as desired.
\end{proof}

\section{Upper Bounding Backlog} \label{sec:upper}

In this section, we prove that for any integer $t \ge 1$, the greedy emptying strategy can guarantee that the backlog does not exceed $O(b(t))$ against \emph{any} filling strategy, after $t$ rounds of the variable-processor cup game, where $b(t)$ is defined as in Equation \eqref{trade-off-curve} in the introduction (Section \ref{sec:intro}). The result can be phrased compactly as follows:

\begin{restatable}{theorem}{thmupper}
\label{thm:main_upper}
    Let $1 \le k \le O(n).$ Then, assuming the emptier follows a greedy emptying strategy in the variable-processor cup game, the filler needs $\Omega\left(k + \frac{k^3}{\log^2 (n/k)}\right)$ rounds to achieve backlog $k$.
\end{restatable}

Our approach to proving the theorem will be as follows.
In Subsection \ref{subsec:l_variant}, we introduce a combinatorial variant of the cup game, which we analyze via potential function arguments. In Subsection \ref{subsec:majorization}, we introduce and prove important results on majorization and domination. Finally, in Subsection \ref{subsec:upper_finish}, we use the results from \ref{subsec:majorization} to justify why this combinatorial variant captures the standard game, and in doing so we prove Theorem \ref{thm:main_upper}.

\subsection{The $\ell$-Variant Cup Game} \label{subsec:l_variant}

Before analyzing the performance of greedy on the standard variable-processor cup game, we consider the following variants of the cup game in which the filler can only make a certain highly structured type of move.

\paragraph{The stone-variant of the cup game.} We define the \emph{stone-variant cup game} as follows (note that this game is equivalent to the ``stone game'' in Section \ref{sec:technical}, but is phrased as a game on cups rather than stones). 

The stone-variant cup game is a variation on the negative-fill variable-processor cup game. 
At the start of each round, the fills are nonnegative integers $x_1 \ge \cdots \ge x_n$, and initially, the fills are all $0$. During every round, the player of the game chooses an integer $k \ge 0$ and an integer $q \ge 1$ such that at least $2q$ of the $x_i$'s are exactly $k$ (if no such $k$ and $q$ exist, then the game is over). Then, among these $x_i$'s, we raise exactly $q$ of them by $1$ and lower exactly $q$ of them by $1$. Finally, at the end of each round (for notational convenience), we re-sort the $x_i$s to be in descending order.  The player's goal is to maximize backlog, that is, the value of $x_1$. 

Let us take a moment to remark on how the stone-variant cup game relates to the negative-fill variable-processor cup game. Recall the proof of Lemma \ref{lem:basic_movement} that, in the negative-fill variable-processor cup game, it is possible for the filler to perform a move with the following net effect: some number $2q$ of cups that all have the same fills $k$ are replaced with $q$ cups that have fills $k + 1/2$ and $ q $ cups that have fills $ k -1/2 $. This is equivalent to performing a move in the stone-variant game, except that fills in the stone-variant game are re-normalized to be twice as large. Thus, the stone-variant game can be viewed as a variation on the negative-fill variable-processor cup game, where the filler's behavior is limited to a specific type of move. 

\paragraph{The $\ell$-variant of the cup game. }
We define the $\ell$-variant of the cup game to be similar to the stone-variant, but in addition, there are infinitely equally spaced checkpoints $0, \ell, 2\ell, \cdots$, where $\ell$ is some positive integer. The $\ell$-variant of the game adds the constraint that: once a cup reaches a checkpoint, it never goes below that checkpoint. In other words, if during some round $k = a \cdot \ell$ for some integer $a$, then we raise $q$ of the $x_i$'s from $k$ to $k+1$ but rather than lowering $q$ other $x_i$'s from $k$ to $k-1,$ we keep them at $k$. (Note that that the checkpoint at $0$ ensures that all cups have non-negative fills.)

\paragraph{How we will use the games. }Although the filler in the two above games has been restricted (compared to the filler in the original game), we will show in subsequent subsections how to reduce the standard variable-processor cup game to the stone-variant, and then how to reduce the stone-variant to the $\ell$-variant. In other words, we show that restricting the filler to make $\ell$-variant-style moves does not make the game (asymptotically) harder for the filler.

In this subsection, we focus on the task of bounding the backlog $\ell$-variant of the game.  As described in Section \ref{sec:technical}, our main insight in analyzing this game is to create two potential functions and compare their growth. In the following propositions, we define two slightly simplified potential functions $\Phi, \Psi$, and analyze how certain operations similar to those in the $\ell$-variant cup game affect $\Phi$ and $\Psi$.

\begin{proposition} \label{prop:potential_1}
    Suppose that $x_1 \ge x_2 \ge \cdots \ge x_n$ are integers, and $q \ge 1$ is an integer such that for some subset $S \subset [n]$ of size $2q$, all of the $x_i$'s for $i \in S$ are all equal. Suppose that for $q$ values of $i \in S$ we increase $x_i$ by $1$, and for the other $q$ values of $i \in S$ we decrease $x_i$ by $1$. Then, the potential function $\Phi = \sum_{i = 1}^{n} x_i^2$ increases by exactly $2q$ and the potential function $\Psi = n \cdot \sum_i |x_i| + \sum_{i < j} |x_i-x_j|$ increases by at least $2q^2.$
\end{proposition}

\begin{proof}
    The fact that $\Phi = \sum x_i^2$ increases by $2q$ is straightforward, since for any $x$, $q \cdot (x+1)^2 + q \cdot (x-1)^2 - 2q \cdot x^2 = 2q,$ and since the $x_i$'s we modified are all initially at the same value.
    
    Now, we prove that $\Psi$ increases by at least $2q^2$. First, we note that since the $x_i$'s are sorted in descending order, we can assume WLOG that $S = [r:s]$ for some $1 \le r \le s \le n$ with $s-r+1 = 2q$. Moreover, we can assume WLOG that we decrease $x_r, \dots, x_{r+q-1}$ by $1$ and increase $x_{r+q}, \dots, x_{s}$ by $1$. Then, we note that if $i, j$ are both not in the set $[r:s],$ that $|x_i-x_j|$ remains constant. Also, if $i$ is not in the set $[r:s],$ then $\sum_{j \in [r:s]} |x_i-x_j|$ does not change if either $x_i \ge x_r+1$ or $x_i \le x_r-1$, and otherwise increases by $2q$ if $x_i = x_r$. Hence, this sum does not decrease for any $i$. Finally, if we just restrict to $i < j$ with $i, j \in [r:s]$, originally each $|x_i-x_j|$ is $0$, but after increasing $q$ of the $x_i$'s by $1$ and decreasing $q$ of the $x_i$'s by $1$, the sum of these $|x_i-x_j|$'s becomes exactly $2q^2$. In total, this means that
\[\sum_{i < j} |x_i-x_j| = \sum_{\substack{i < j \\ i, j \not\in [r:s]}} |x_i-x_j| +  \sum_{i \not\in [r:s]} \sum_{j \in [r:s]} |x_i-x_j| + \sum_{\substack{i < j \\ i, j \in [r:s]}} |x_i-x_j|\]
    increases by at least $0 + 0 + 2q^2 = 2q^2$. Finally, we have that $\sum_i |x_i|$ does not decrease, since for any integer $x$, $q \cdot |x+1| + q \cdot |x-1| \ge 2q \cdot |x|$. This completes the proof.
\end{proof}

\begin{proposition} \label{prop:potential_2}
    Suppose that $x_1 \ge x_2 \ge \cdots \ge x_n \ge 0$ are integers, such that for some subset $S \subset [n]$ of size $q$, all of the $x_i$'s for $i \in S$ are all $0$. Suppose that we increase $x_i$ by $1$ for all $i \in S$. Then, the potential function $\Phi = \sum_{i = 1}^{n} x_i^2$ increases by exactly $q$ and the potential function $\Psi = n \cdot \sum_i |x_i| + \sum_{i < j} |x_i-x_j|$ increases by at least $q^2.$
\end{proposition}

\begin{proof}
    Showing that $\Phi$ increases by $q$ is straightforward, since each $x_i$ for $i \in S$ increases from $0$ to $1$.
    
    Now, we show that $\Psi$ increases by at least $q^2.$ First, note that since the $x_i$'s are sorted in descending order, we can assume WLOG that $S = [r:n]$ for $r = n-q+1$. Now, note that $n \cdot \sum_i |x_i|$ increases by exactly $n \cdot q$ since each of $x_r, \dots, x_n$ increase from $0$ to $1$. In addition, for any $i < j,$ $|x_i-x_j|$ remains constant if $i, j < r$ or $i, j \ge r$, and does not decrease by more than $1$ if $i < r \le j$. There are exactly $(r-1) \cdot (n-r+1) = (n-q) \cdot q$ pairs $(i, j)$ satisfying $i < r \le j$, so $\Psi$ increases by at least $n \cdot q - (n-q) \cdot q \ge q^2$.
\end{proof}

Now, we return to the $\ell$-variant cup game. We think of $\textbf{x}^{(t)} = (x_1^{(t)}, \dots, x_n^{(t)})$ as the set of fills after time step $t$, where $\textbf{x}^{(0)} = (0, 0, \dots, 0)$. Suppose that there exists a sequence $\textbf{x}^{(0)}, \textbf{x}^{(1)}, \dots, \textbf{x}^{(T)}$ such that after $T$ time steps, there are precisely $n_a$ cups that have reached the checkpoint $a \cdot \ell$ for each $a \ge 0$. (In other words, $x_i^{(T)} \ge a \cdot \ell$ if and only if $i \le n_a$.) Note that $n_0 = n$, and since a cup never goes below a checkpoint after reaching it, if $i > n_a$ then $x_i^{(t)} < a \cdot \ell$ for all $0 \le t \le T.$

Now, we will consider the following potential functions that are modifications of the potential functions $\Phi, \Psi$ that we used in Propositions \ref{prop:potential_1} and \ref{prop:potential_2}. For any nonnegative integer $a$, we define $f_{a, \ell}(x) := \max(0, \min(x-a \cdot \ell, \ell)).$ In other words,
\[f_{a,\ell}(x) = \begin{cases} 0 & \text{if } x < a \cdot \ell \\ x-a \cdot \ell & \text{if } a \cdot \ell \le x \le (a+1) \cdot \ell \\ \ell & \text{if } x > (a+1) \cdot \ell .\end{cases}\]
For any integers $x_1, \dots, x_n$ and $0 \le n_a \le n$, define $S := S_{n_a}(x_1, \dots, x_n)$ to be the set of the largest $n_a$ indices of the $x_i$'s (breaking ties arbitrarily). In other words, $S = S_{n_a}(x_1, \dots, x_n) \subset [n]$ is a set of size $n_a$, such that for all $i \in S$ and $j \not\in S$, $x_i \ge x_j$.
For $S = S_{n_a}(x_1, \dots, x_n)$, let the potential function $\Phi_a$ be
\[\Phi_a := \Phi_{a, \ell, n_a}(x_1, \dots, x_n) = \sum_{i \in S} f_{a, \ell}(x_i)^2\]
and let the potential function $\Psi_a$ be
\[\Psi_a := \Psi_{a, \ell, n_a}(x_1, \dots, x_n) = n_a \cdot \sum_{i \in S} |f_{a, \ell}(x_i)| + \sum_{\substack{i < j \\ i, j \in S}} \left|f_{a, \ell}(x_i) - f_{a, \ell}(x_j)\right|.\]
Note that $\Phi_a$ and $\Psi_a$ are defined for each nonnegative integer $a$, and that the tiebreaking order to determine $S$ does not affect the values $\Phi_a$ or $\Psi_a$. In addition, note that the numbers $n_a$ only depend on the final state $\textbf{x}^{(T)}$, whereas the potential functions are defined for each $\textbf{x}^{(t)}$ for all $0 \le t \le T$.

\begin{lemma} \label{lem:inequality}
    Let $k \le n$, $h = \Theta\left(\log\left(\frac{n}{k} + 1\right)\right)$ such that $h \ge 2$ and $h|k$, and $\ell = \frac{k}{h}$. Suppose that in the $\ell$-variant cup game, it is possible that after $T$ time steps, there are exactly $n_a$ cups that have reached the $a^{\text{th}}$ checkpoint or above, for integers $n_0, n_1, n_2, \dots$ Moreover, suppose that $n_h \ge k$. Then, $T \ge \Omega(h \cdot \ell^3)$.
\end{lemma}

\begin{proof}
    By definition of the $n_a$'s, we have that $n = n_0 \ge n_1 \ge \cdots \ge n_h,$ and we are assuming $n_h \ge k.$ Because $\frac{n_0}{n_h} \le \frac{n}{k} \le e^{C \cdot h}$ for some fixed constant $C > 0$, we must have that for at least $\frac{h}{2}$ choices of $0 \le a \le h-1,$ that $\frac{n_a}{n_{a+1}} \le e^{2C},$ or equivalently, $\frac{n_{a+1}}{n_a} \ge e^{-2C} =: c$ for some fixed $c > 0$.
    
    Let $T_a$ represent the set of time steps $t \le T$
    during which cups with fill between $a \cdot \ell$ and $(a+1) \cdot \ell-1$, inclusive, at the beginning of the time step are modified. Then, for any time step $t \in T_a$, the potential function $\Phi_a$ increases by some $q_t$ and the potential function $\Psi_a$ increases by at least $q_t^2/2$, by Propositions \ref{prop:potential_1} and \ref{prop:potential_2}. Otherwise, if $t \not\in T_a$, then $\Phi_a$ and $\Psi_a$ are unaffected: to see why, we justify that for all $i$, $f_{a, \ell}(x_i)$ is unaffected at time step $t$ if $t \not\in T_a$. Note that we are only modifying cups of fill either at least $(a+1) \cdot \ell$ or at most $a \cdot \ell-1$. If we modify a cup of fill $x_i \ge (a+1) \cdot \ell$, the cup will not go below the $(a+1)^{\text{th}}$ checkpoint, so $f_{a, \ell}(x_i)$ will remain $\ell$. Likewise, if we modify a cup of fill $x_i \le a \cdot \ell-1$, then $x_i \le a \cdot \ell$ even if we increase the fill by $1$, so $f_{a, \ell}(x_i) = 0$ still. Finally, note that sorting the cups based on $x_i$ has no effect on $\Phi_a$ or $\Psi_a$. We also remark that each time $t \le T$ is in exactly one set $T_a$ for $a \ge 0$, since every cup we modify at any time step $t$ had the same fill at the start of time step $t$.
    
    Now, fix some $a$ such that $\frac{n_{a+1}}{n_a} \ge c.$ At the beginning of the game, we have that $f_{a, \ell}(x_i) = 0$ for all $i$, which means $\Phi_a = \Psi_a = 0$. Moreover, at the end of the $T$ time steps, $f_{a, \ell}(x_i) \in [0, \ell]$ for all of the top $n_a$ values of $x_i$, with at least $n_{a+1} \ge c \cdot n_a$ of them being $\ell$. Therefore, at the end of the game, $\Phi_a \ge n_{a+1} \cdot \ell^2 \ge c n_a \cdot \ell^2$ but $\Psi_a \le n_a^2 \cdot \ell + {n_a \choose 2} \cdot \ell \le 2 n_a^2 \cdot \ell$. In addition, recall that for each $t \in T_a$, $\Phi_a$ increases by some $q_t$, but $\Psi_a$ increases by at least $q_t^2/2.$ Therefore,
\[\sum_{t \in T_a} q_t = \Phi_a \ge cn_a \cdot \ell^2\]
    but
\[\sum_{t \in T_a} q_t^2 \le 2\Psi_a \le 4n_a^2 \cdot \ell.\]
    Now, using the Cauchy-Schwarz inequality, we have that
\[|T_a| \ge \frac{\left(\sum_{t \in T_a} q_t\right)^2}{\sum_{t \in T_a} q_t^2} \ge \frac{(c n_a \cdot \ell^2)^2}{4 n_a^2 \cdot \ell} = \frac{c^2}{4} \cdot \ell^3.\]

    Finally, since the $T_a$'s are disjoint sets, and since this inequality is true for at least $\frac{h}{2}$ values of $a$, we have that $T \ge \sum_a |T_a| \ge \frac{c^2}{8} \cdot h \cdot \ell^3,$ as desired.
\end{proof}

Note that Lemma \ref{lem:inequality} had an assumption that $n_h \ge k$, or that at least $k$ cups have reached the checkpoint at $h \cdot \ell$. The following proposition will be useful in removing that assumption.

\begin{proposition} \label{prop:no_gaps}
    In the $\ell$-variant cup game, if there exists a cup of fill $k$ for some $k \ge 3$, then there must exist a cup of fill either $k-1$ or $k-2$.
\end{proposition}

\begin{proof}
    Suppose not, so there was a first time $t$ such that for some integer $k$, there was a cup of fill $k$ but no cups of fill $k-1, k-2$. Suppose at the previous time step there was a cup of fill $k-2$. Then, at the current time step, either some of these cups still have fill $k-2$, or at least one of the cups increased to have fill $k-1$. Suppose at the previous time step, some of the cups had fill $k-1$. Then, at the current time step, either some of the cups still have fill $k-1$, or we increased half of these cups to fill $k$ and decreased the other half to fill $k-2$. It is possible that $k-1$ was a checkpoint, in which case the decreased cups would remain at fill $k-1$. Regardless, if there was a cup of fill either $k-1$ or $k-2$ at the previous time step, there will remain at least one such cup.
    
    Therefore, in the previous time step there were no cups of fill $k - 1$ or $k - 2$. On the other hand, there must have been at least one cup of fill $k$ or greater, since in the current time step there is a cup with fill $k$. Thus, in time step $t - 1$ there is some fill $k' \ge k$ such that there was a cup with fill $k'$ but no cups with fill $k' - 1$ or $k' - 2$. This contradicts the assumption that $t$ is the smallest time such that for some integer $k$, there is a cup of fill $k$ but no cups of fill $k-1, k-2$.
\end{proof}

Combining Proposition \ref{prop:no_gaps} with Lemma \ref{lem:inequality}, we have the following result, which essentially provides us with tight bounds for the $\ell$-variant cup game.

\begin{theorem} \label{thm:inequality}
    Let $k \le n$, $h = \Theta\left(\log\left(\frac{n}{k} + 1\right)\right)$ such that $2 \le h \le k$, and $\ell = \left\lfloor\frac{k}{h}\right\rfloor$. Suppose that in the $\ell$-variant cup game, it is possible that after $T$ time steps, there is a cup of fill $4 k$ or greater. Then, $T \ge \Omega(h \cdot \ell^3) = \Omega(k^3/h^2)$.
\end{theorem}

\begin{proof}
    By Proposition \ref{prop:no_gaps}, if there is a cup of fill $4 k$ or greater, then there must be at least $k$ cups of fill $k$ or greater, meaning that there are at least $k$ cups which have reached the $h^{\text{th}}$ checkpoint or higher, since $k \ge h \cdot \ell$. Therefore, we can apply Lemma \ref{lem:inequality} to conclude our proof.
\end{proof}

\subsection{Majorization and Domination} \label{subsec:majorization}

In this subsection, we develop the necessary toolkit that will allow us to relate the $\ell$-variant cup game to the standard variable-processor cup game. First, we define majorization and domination, which relate to comparing two real-valued sequences, and we state some basic results about them. We then prove more complicated results which will be important in relating the two games.

For any sequence of fills $x_1, \dots, x_n \in \BR$ of the $n$ cups, we abbreviate the state of the cups as $\{x_i\}$. We will explicitly specify whenever we assume $x_1 \ge x_2 \ge \cdots \ge x_n$, in which case we say that the sequence $\{x_i\}$ is \defnn{sorted}. We also do not require $x_1, \dots, x_n$ to be nonnegative in this subsection.

\begin{definition} \label{def:majorization}
    Suppose that we have two sequences $x_1 \ge x_2 \ge \cdots \ge x_n$ and $y_1 \ge y_2 \ge \cdots \ge y_n$, where $\sum_{i = 1}^{n} x_i = \sum_{i = 1}^{n} y_i.$ 
    Then, we say that $\{x_i\}$ \defnn{majorizes} $\{y_i\}$ if for every $m \le n,$ $\sum_{i = 1}^{m} x_i \ge \sum_{i = 1}^{m} y_i$.
    In general, if the sequences $\{x_i\}$ and $\{y_i\}$ are not necessarily sorted, we say that $\{x_i\}$ majorizes $\{y_i\}$ if this is true after sorting $\{x_i\}$ and $\{y_i\}$.
\end{definition}

\begin{remark}
    Note that majorization satisfies transitivity, so for any real number $r$, majorization creates a partial order on the set of length $n$ sequences (up to permutation of terms) that add to $r$.
\end{remark}

We note two basic results about majorization.

\begin{proposition} \label{prop:ez2}
    Suppose that $\{a_1, \dots, a_k\}$ majorizes $\{b_1, \dots, b_k\}$ (where $\sum_{i = 1}^{k} a_i = \sum_{i = 1}^{k} b_i$). Then, for any sequence $\{c_1, \dots, c_r\},$ $\{a_1, \dots, a_k, c_1, \dots, c_r\}$ majorizes $\{b_1, \dots, b_k, c_1, \dots, c_r\}$.
\end{proposition}

\begin{proposition} \label{prop:perturbation}
    If $\{x_i\}$ and $\{y_i\}$ are sorted, and $\{x_i\}$ majorizes $\{y_i\}$, then we can convert $\{y_i\}$ into $\{x_i\}$ via a finite series of perturbations of increasing some $y_j$ by $\eps$ and decreasing $y_k$ by $\eps$ for some $0 < \eps < 1$ and $j > k$. (The perturbations do not all need to use the same value of $\eps$.) Moreover, after each such perturbation, the sequence is still sorted.
\end{proposition}

\begin{definition} \label{def:domination}
    Given two sequences $x_1 \ge x_2 \ge \cdots \ge x_n$ and $y_1 \ge y_2 \ge \cdots \ge y_n,$ we say that $\{x_i\}$ \defnn{dominates} $\{y_i\}$ if $x_i \ge y_i$ for all $i$.
    In general, we say that $\{x_i\}$ dominates $\{y_i\}$ if this is true after sorting $\{x_i\}$ and $\{y_i\}$.
\end{definition}

\begin{remark}
    Unlike majorization, where we only say that $\{x_i\}$ majorizes $\{y_i\}$ if $\sum_{i=1}^n x_i = \sum_{i=1}^n y_i$, we only have that $\sum_{i=1}^n x_i = \sum_{i=1}^n y_i$ and $\{x_i\}$ dominates $\{y_i\}$ if the sequences $\{x_i\}$ and $\{y_i\}$ are the same (after sorting).
\end{remark}

We note the following well-known proposition about domination.

\begin{proposition} \label{prop:super_dom_ez}
    Suppose that $x_1, \dots, x_n$ and $y_1, \dots, y_n$ are real-valued, not necessarily sorted, sequences, but where $x_i \ge y_i$ for all $i$. Then, after sorting, $\{x_i\}$ dominates $\{y_i\}.$
\end{proposition}

For completeness, we prove Propositions \ref{prop:ez2}, \ref{prop:perturbation}, and \ref{prop:super_dom_ez} in Appendix \ref{sec:appendix_ez}.

Next, we prove the following proposition relating to adding fills.

\begin{proposition} \label{prop:ez}
    If the cups have fills $x_1 \ge x_2 \ge \cdots \ge x_n$ before the start of a time step, and if the filler fills $a_i$ to each $x_i$ during the time step, where $0 \le a_i \le 1$, then we can always assume without loss of generality that $a_1+x_1 \ge a_2+x_2 \ge \cdots \ge a_n+x_n$.
\end{proposition}

\begin{proof}
    Suppose that $a_i+x_i \le a_{i+1}+x_{i+1}$. Then, since $x_i \ge x_{i+1}$, this means that $x_{i+1} \le x_i \le a_i+x_i \le a_{i+1}+x_{i+1} \le 1 + x_{i+1}$ since fills are always between $0$ and $1$. Instead, we can replace $a_{i+1}$ with $(a_i+x_i)-x_{i+1}$ and $a_i$ with $(a_{i+1}+x_{i+1})-x_i$, which are both in the range $[0, 1]$, and the fills of the $i^{\text{th}}$ and $(i+1)^{\text{th}}$ cups will just be switched. We can repeat this bubble-sorting procedure until we actually have $a_1+x_1 \ge a_2+x_2 \ge \cdots \ge a_n+x_n,$ and we do not affect the fills of the cups up to permutation. Moreover, all fills are still between $0$ and $1$ for that round.
\end{proof}

Hence, we assume that for any round, both before and after the filler moves, the cups are sorted in non-increasing order. Next, the (greedy) emptier removes $1$ unit of water from each of the first $p$ cups, and then sorts the cups in non-increasing order.

Our main technical result to reduce from the standard variable-processor cup game to the stone-variant cup game is the following lemma. The lemma roughly implies that if some sequence $\{x_i\}$ majorizes another sequence $\{y_i\}$, then the filler would always be at an advantage in the future (in the negative-fill cup game) if the fills were $\{x_i\}$ instead of $\{y_i\}$.

\begin{lemma} \label{lem:majorization}
    Suppose that $x_1 \ge x_2 \ge \cdots \ge x_n$ and $y_1 \ge y_2 \ge \cdots \ge y_n$ are reals, where $\{x_i\}$ majorizes $\{y_i\}$, and that there exists some filling procedure on $\{y_i\}$ such that after greedy emptying (allowing for negative fills), $\{y_i\}$ becomes some sequence $\{y_i'\}$ ($\{y_i'\}$ is not necessarily sorted). Then, there exists a filling procedure on $\{x_i\}$ (with a possibly different value of $p$) such that after greedy emptying (allowing for negative fills), $\{x_i\}$ becomes $\{x_i'\}$ and $\{x_i'\}$ majorizes $\{y_i'\}$.
\end{lemma}

\begin{proof}
    By Proposition \ref{prop:perturbation} and the transitivity of majorization, we can assume that the only difference between $\{x_i\}$ and $\{y_i\}$ is that there exist some $j < k$ such that $x_j = y_j+\eps$ and $x_k=y_k-\eps,$ and for all other $i$, $x_j = y_j$, and moreover, the sequences $\{x_i\}$ and $\{y_i\}$ are sorted.
    
    By Proposition \ref{prop:ez}, we may assume in the filling procedure for $\{y_i\}$ we add $a_i$ to each $y_i$, so that $\sum a_i = p$ (where $1 \le p \le n$ is an integer) and $a_1+y_1 \ge \dots \ge a_n+y_n,$ and then remove $1$ from each of the first $p$ of the $(a_i+y_i)$'s (as they are the largest). 
    
    First, if $p = n$, it is clear that the filler fills every cup with $1$ unit, and then the emptier empties every cup, so $\{y_i'\} = \{y_i\}$. So, we can just do the same filling procedure on $\{x_i\}$, and after the emptier moves, we have $\{x_i'\} = \{x_i\}$ which majorizes $\{y_i\}$. So, from now on we assume $1 \le p \le n-1$. We have $4$ main cases for what can happen.
    
    \textbf{Case 1:} $a_j > 0$ and $a_k < 1$. In this case, let $\eps' = \min(a_j, 1-a_k, \eps).$ If we had $x_j = y_j+\eps'$ and $x_k = y_k-\eps'$ instead, then we could just perform the following filling procedure on $\{x_i\}$ to get $\{x_i'\}$: fill the $j^{\text{th}}$ cup with $a_j-\eps'$, the $k^{\text{th}}$ cup with $a_k+\eps'$, and the remaining cups with $a_i$ for each $i$ (note that, by the definition of $\eps'$, this is guaranteed to be a valid filling procedure, placing between $0$ and $1$ units of water in each cup). Then, in fact we would have $\{x_i'\} = \{y_i'\},$ since even before the greedy emptying (but after the filling), the fills would be the same. If $\eps = \eps'$, this completes the case. Otherwise, either $a_j = \eps'$ or $1-a_k = \eps'$, in which case we let $\{z_i\}$ be the same as $\{y_i\}$, except $z_j = y_j+\eps'$ and $z_k = y_k-\eps'$. We have shown there is a filling procedure on $\{z_i\}$ that obtains the same results as the filling procedure on $\{y_i\}$, so it suffices to prove the lemma for $\{x_i\}$ and $\{z_i\}$ (with the filler placing water into $\{z_i\}$ so that after greedy emptying we get $\{y_i'\}$) instead of for $\{x_i\}$ and $\{y_i\}$. But in this case, since $\eps' \in \{a_j, 1-a_k\}$, the filler on $\{z_i\}$ either does not place any fill in the $j^{\text{th}}$ cup (if $\eps' = a_j$) or places a full $1$ unit of fill into the $k^{\text{th}}$ cup (if $\eps' = 1 - a_k$), so we have successfully reduced to the case where either $a_j = 0$ or $a_k = 1.$
    
    \textbf{Case 2:} $a_j = 0$ and $p < j$. In this case, for the $\{x_i\}$'s we do the same filling procedure of adding $a_i$ to each $x_i$. Therefore, since $a_1+y_1 \ge a_2+y_2 \ge \dots \ge a_n+y_n$ and since $x_j = y_j+\eps, x_k = y_k-\eps$, the largest $p$ values of $\{x_i+a_i\}$ are either the first $p$ values, or the first $p-1$ values and the $j^{\text{th}}$ value. But since $a_j = 0$, we have that $x_j+a_j = x_j,$ and since both the $\{y_i\}$'s and $\{x_i\}$'s are sorted in nonincreasing order, this means that $x_j+a_j = x_j \le x_p \le x_p+a_p$, so the largest $p$ values of $\{x_i+a_i\}$ are just the first $p$ values. So, the greedy emptier empties the first $p$ cups, and $y_i' = x_i'$ for all $i$, except $x_j' = y_j+a_j+\eps$ and $x_k' = y_k+a_k-\eps$. Since $y_j+a_j \ge y_k+a_k,$ this implies that $\{x_j', x_k'\}$ majorizes $\{y_j', y_j'\}$, so by Proposition \ref{prop:ez2}, $\{x_i'\}$ majorizes $\{y_i'\}.$
    
    \textbf{Case 3:} $a_j = 0$ and $j \le p$. In this case, for the $\{x_i\}$'s we fill $1$ unit for cup $j$ and $a_i$ units for all $i \neq j$ (so we now have filled $p+1$ total units). But now, we empty the fullest $p+1$ cups (instead of the fullest $p$ cups). Note that, whereas the transformation of $\{y_i\}$ to $\{y_i'\}$ uses a $p$-processor time step of the cup game, the transformation from $\{x_i\}$ to $\{x_i'\}$ uses a $p + 1$-processor time step of the cup game.
    We have $5$ possible subcases:
    
\begin{enumerate}[label=\textbf{\alph*})]
    \item $j \le p < k-1$. In this case, after filling the $\{x_i\}$ cups, the $p+1$ fullest cups are just the first $p+1$ cups. So, the only cups that are of different fills when changing from $\{y_i'\}$ to $\{x_i'\}$ are 
\begin{alignat*}{5}
    y_j' &= y_j-1 && \hspace{0.25cm} \text{becomes} \hspace{0.25cm} && x_j' &&= y_j+\eps && \hspace{0.5cm} (\text{since } a_j = 0 \text{ and } x_j = y_j+\eps) \\
    y_{p+1}' &= y_{p+1}+a_{p+1} && \hspace{0.25cm} \text{becomes} \hspace{0.25cm} && x_{p+1}' &&= y_{p+1}+a_{p+1}-1 \\
    y_k' &= y_k+a_k && \hspace{0.25cm} \text{becomes} \hspace{0.25cm} && x_k' &&= y_k+a_k-\eps && \hspace{0.5cm} (\text{since } x_k = y_k-\eps)
\end{alignat*}
    Note that $y_j'+y_{p+1}'+y_k' = x_j'+x_{p+1}'+x_k'$, and since $y_j \ge y_{p+1}+a_{p+1} \ge y_k+a_k$, we have that $\max(y_j', y_{p+1}', y_k') \le y_j \le x_j' \le \max(x_j', x_{p+1}', x_k')$, and $\min(y_j', y_{p+1}', y_k') = \min(y_j-1, y_k+a_k) \ge \min(y_{p+1}+a_{p+1}-1, y_k+a_k-\eps) \ge \min(x_j', x_{p+1}', x_k').$ Therefore, $\{x_j', x_{p+1}', x_k'\}$ majorizes $\{y_j', y_{p+1}', y_k'\}$, so by Proposition \ref{prop:ez2}, $\{x_i'\}$ majorizes $\{y_i'\}$.
    
    \item $p = k-1$ and either $y_k+a_k-\eps \ge y_{k+1}+a_{k+1}$ or $k = n$. In this case, after filling the $\{x_i\}$ cups, the $p+1$ fullest cups are still the first $p+1$ cups. So, the only cups that are of different fills when changing from $\{y_i'\}$ to $\{x_i'\}$ are
\begin{alignat*}{5}
    y_j' &= y_j-1 && \hspace{0.25cm} \text{becomes} \hspace{0.25cm} && x_j' &&= y_j+\eps &&\hspace{0.5cm} (\text{since } a_j = 0) \\
    y_k' &= y_k+a_k && \hspace{0.25cm} \text{becomes} \hspace{0.25cm} && x_k' &&= y_k+a_k-\eps-1 &&\hspace{0.5cm} (\text{since } k = p+1).
\end{alignat*}
    We have that $y_j'+y_k' = x_j'+x_k'$, but $x_j' \ge y_j = y_j + a_j \ge \max(y_j', y_k + a_k) = \max(y_j', y_k'),$ so $\{x_j', x_k'\}$ majorizes $\{y_j', y_k'\}.$ Therefore, by Proposition \ref{prop:ez2}, $\{x_i'\}$ majorizes $\{y_i'\}$.

    \item $p = k-1$ and $y_k+a_k-\eps < y_{k+1}+a_{k+1}$. In this case, after filling the $\{x_i\}$ cups, the $p+1$ fullest cups are the first $p$ cups and the $(k+1)^{\text{st}}$ cup. So, the only cups that are of different fills when changing from $\{y_i'\}$ to $\{x_i'\}$ are
\begin{alignat*}{5}
    y_j' &= y_j-1 && \hspace{0.25cm} \text{becomes} \hspace{0.25cm} && x_j' &&= y_j+\eps &&\hspace{0.5cm} (\text{since } a_j = 0) \\
    y_k' &= y_k+a_k && \hspace{0.25cm} \text{becomes} \hspace{0.25cm} && x_k' &&= y_k+a_k-\eps && \\
    y_{k+1}' &= y_{k+1}+a_{k+1} && \hspace{0.25cm} \text{becomes} \hspace{0.25cm} && x_{k+1}' &&= y_{k+1}+a_{k+1}-1 &&
\end{alignat*}
    Note that $y_j'+y_k'+y_{k+1}' = x_j'+x_k'+x_{k+1}'$. Also, $\max(y_j', y_k', y_{k+1}') \le y_j +a_j = y_j \le x_j' \le \max(x_j', x_k', x_{k+1}')$, and $\min(y_j', y_k', y_{k+1}') \ge \min(y_j-1, y_{k+1}+a_{k+1}) \ge x_{k+1}' \ge \min(x_j', x_k', x_{k+1}')$. Therefore, $\{x_j', x_k', x_{k+1}'\}$ majorizes $\{y_j', y_k', y_{k+1}'\}$, so by Proposition \ref{prop:ez2}, $\{x_i'\}$ majorizes $\{y_i'\}$.

    \item $p \ge k$ and either $y_k+a_k-\eps \ge y_{p+2}+a_{p+2}$ or $p = n-1$. In this case, after filling the $\{x_i\}$ cups, the $p+1$ fullest cups are the first $p+1$ cups. So, the only cups that are of different fills when changing from $\{y_i'\}$ to $\{x_i'\}$ are
\begin{alignat*}{5}
    y_j' &= y_j-1 && \hspace{0.25cm} \text{becomes} \hspace{0.25cm} && x_j' &&= y_j+\eps &&\hspace{0.5cm} (\text{since } a_j = 0) \\
    y_k' &= y_k+a_k-1 && \hspace{0.25cm} \text{becomes} \hspace{0.25cm} && x_k' &&= y_k+a_k-1-\eps && \\
    y_{p+1}' &= y_{p+1}+a_{p+1} && \hspace{0.25cm} \text{becomes} \hspace{0.25cm} && x_{p+1}' &&= y_{p+1}+a_{p+1}-1 &&
\end{alignat*}

    Note that $y_j'+y_k'+y_{p+1}' = x_j'+x_k'+x_{p+1}'.$ Also, since $y_j \ge y_k+a_k \ge y_{p+1}+a_{p+1},$ we have that $\max(x_j', x_k', x_{p+1}') = y_j+\eps > \max(y_j', y_k', y_{p+1}')$. Also, since $x_k' \le y_k' \le y_j'$ and $x_{p+1}' \le y_{p+1}'$, we have that $\min(x_j', x_k', x_{p+1}') \le \min(y_j', y_k', y_{p+1}')$. Therefore, $\{x_j', x_k', x_{p+1}'\}$ majorizes $\{y_j', y_k', y_{p+1}'\}$, so by Proposition \ref{prop:ez2}, $\{x_i'\}$ majorizes $\{y_i'\}$.

    \item $p \ge k$ and $y_k+a_k-\eps < y_{p+2}+a_{p+2}$. In this case, after filling the $\{x_i\}$ cups, the $p+1$ fullest cups are the first $p+2$ cups, except the $k^{\text{th}}$ cup. So, the only cups that are of different fills when changing from $\{y_i'\}$ to $\{x_i'\}$ are
\begin{alignat*}{5}
    y_j' &= y_j-1 && \hspace{0.25cm} \text{becomes} \hspace{0.25cm} && x_j' &&= y_j+\eps &&\hspace{0.5cm} (\text{since } a_j = 0) \\
    y_k' &= y_k+a_k-1 && \hspace{0.25cm} \text{becomes} \hspace{0.25cm} && x_k' &&= y_k+a_k-\eps && \\
    y_{p+1}' &= y_{p+1}+a_{p+1} && \hspace{0.25cm} \text{becomes} \hspace{0.25cm} && x_{p+1}' &&= y_{p+1}+a_{p+1}-1 && \\
    y_{p+2}' &= y_{p+2}+a_{p+2} && \hspace{0.25cm} \text{becomes} \hspace{0.25cm} && x_{p+2}' &&= y_{p+2}+a_{p+2}-1 &&
\end{alignat*}

    We have that $y_j'+y_k'+y_{p+1}'+y_{p+2}' = x_j'+x_k'+x_{p+1}'+x_{p+2}'$. Also, since $y_j \ge y_k+a_k \ge y_{p+1}+a_{p+1} \ge y_{p+2}+a_{p+2}$, we have that $x_j' \ge x_k' \ge x_{p+1}' \ge x_{p+2}'$. However, it is simple to see that $x_j' \ge \max(y_j', y_k', y_{p+1}', y_{p+2}')$, that $x_{p+2}' \le \min(y_j', y_k', y_{p+1}', y_{p+2}')$, and that $x_j'+x_k' = y_j+(y_k+a_k)$ which is greater than or equal to the sum of any two of $y_j', y_k', y_{p+1}', y_{p+2}'$. Therefore, $\{x_j', x_k', x_{p+1}', x_{p+2}'\}$ majorizes $\{y_j', y_k', y_{p+1}', y_{p+2}'\},$ so by Proposition \ref{prop:ez2}, $\{x_i'\}$ majorizes $\{y_i'\}$.
    
\end{enumerate}
    \textbf{Case 4:} $a_k = 1$.
We show that this case can be derived from Cases 2 and 3. The basic idea is that we proceed in three steps: first we make use of symmetry within the negative-fill variable-processor cup game in order to show that if we began with fills $\{-y_1, -y_2, \ldots, -y_n\}$, it would be possible to perform a round in the game that results in states $\{-y_1', -y_2', \ldots, -y_n'\}$; we then apply Cases 2 and 3 to the states $\{-x_1, -x_2, \ldots, -x_n\}$ and $\{-y_1, -y_2, \ldots, -y_n\}$ in order to deduce that it is possible to transform $\{-x_1, -x_2, \ldots, -x_n\}$ in a single round into some $\{-x_1', -x_2', \ldots, -x_n'\}$ that majorizes $\{-y_1', -y_2', \ldots, -y_n'\}$; finally, we again make use of symmetry to argue that it is also possible to transform $\{x_1, x_2, \ldots, x_n\}$ into $\{x_1', x_2', \ldots, x_n'\}$ in a single round, and that $\{x_1', x_2', \ldots, x_n'\}$ majorizes $\{y_1', y_2', \ldots, y_n'\}$. 
    
    First, note that by the definition of majorization, $\{x_1, \dots, x_n\}$ majorizes $\{y_1, \dots, y_n\}$ if and only if $\{-x_n, \dots, -x_1\}$ majorizes $\{-y_n, \dots, -y_n\}$.
    Therefore, since $\{x_i\}$ and $\{y_i\}$ are sorted, we have that $\{-x_n, \dots, -x_1\}$ majorizes $\{-y_n, \dots, -y_n\}$ with $-x_k = -y_k+\eps$ and $-x_j = -y_j-\eps$.
    
    Now, instead of adding $a_i$ to each $y_i$ (where $a_k = 1$), we consider what happens if we add $1 - a_{i}$ to $-y_i$ for each $i$. This results in fills $1-(y_n+a_n) \ge 1-(y_{n-1}+a_{n-1}) \ge \cdots \ge 1-(y_1+a_1).$ Then, since $\sum_{i = 1}^{n} (1-a_i) = n-p$, the greedy emptier subtracts $1$ from the $n-p$ fullest cups to get the sequence $\{-(y_n-a_n), \dots, -(y_{p+1}-a_{p+1}), 1-(y_p+a_p), \dots, 1-(y_1+a_1)\}$. This sequence precisely equals $\{-y_n', -y_{n-1}', \dots, -y_1'\}$, where $y_i' = y_i+a_i-1$ for $1 \le i \le p$ and $y_i' = y_i+a_i$ otherwise. (Note that $\{y_1', \dots, y_n'\}$ is not necessarily sorted.) In other words, it is possible to perform a round on the $\{-y_i\}$s that brings us directly to the $\{-y_i'\}$s.
    
    Since $a_k = 1$, we have that $1-a_k = 0$, and moreover, $-y_k$ comes before $-y_j$ in the non-increasing sorted order. Thus, we may use either Case $2$ or Case $3$ on the sequences $\{-x_n, \dots, -x_1\}$ and $\{-y_n, \dots, -y_1\}$ to obtain that there exists a sequence of fills $1-b_{n}, \dots, 1-b_1$ for $0 \le b_1, \dots, b_n \le 1$, such that $\sum_{i = 1}^{n} (1-b_i) = n-p'$ for some $p'$ possibly different from $p$, and such that when $1-b_i$ is added to $-x_i$ for all $i$ and then $1$ is subtracted from the fullest $n-p+1$ cups, the resulting sequence majorizes $\{-y_n', \dots, -y_1'\}$. Notice, however, that this resulting sequence is just $\{1-(b_n+x_n), \dots, 1-(b_1+x_1)\}$ but with the largest $n-p'$ values lowered by $1$; or equivalently, it is $\{-(b_n+x_n), \dots, -(b_1+x_1)\}$ but with the smallest $p'$ values increased by $1$. So, returning to the original sequence $\{x_1, \dots, x_n\}$, we have shown that there exist  $0 \le b_1, \dots, b_n \le 1$,  where $\sum_{i = 1}^{n} b_i = p'$, such that, if the filler added $b_i$ to each $x_i$, and the emptier subtracted $1$ from the largest $p'$ values, we would get a sequence $\{x_1', \dots, x_n'\}$, not necessarily sorted, such $\{-x_n', \dots, -x_1'\}$ majorizes $\{-y_n', \dots, -y_1'\}$ after sorting. Since majorization is preserved by negation, it follows that $\{x_1', \dots, x_n'\}$ majorizes $\{y_1', \dots, y_n'\}$, as desired.
\end{proof}






We next prove a similar result to Lemma \ref{lem:majorization}, though for domination and specifically for a time step of the stone-variant cup game.

\begin{lemma} \label{lem:domination}
    Suppose that $\{x_i\}$ and $\{y_i\}$ are sorted and integer valued, and $\{x_i\}$ dominates $\{y_i\}$. Suppose we perform a round of the stone-variant cup game on $\{y_i\}$ to get some sequence $\{y_i'\}.$ Then, one can perform a round of the stone-variant cup game on $\{x_i\}$ to get a sequence $\{x_i'\}$ that dominates $\{y_i'\}$.
\end{lemma}

\begin{proof}
    Suppose that $y_1 \ge y_2 \ge \cdots \ge y_n,$ where for some integer $k$, $[p:q] = \{p, p+1, \dots, q\}$ is the set of indices $i$ such that $y_i = k$. Assume that for the stone-variant cup game on $\{y_i\}$, we increase $y_i$ for $i \in [p:p']$ by $1$, and decrease $y_i$ for $i \in [q':q]$ by $1$, where $p' < q'$ and $p'-p = q-q' \ge 0.$ We can always make the assumption WLOG that any given time step of the stone-variant cup game proceeds in this fashion. In addition, note that the resulting sequence $\{y_i'\}$ is already sorted.
    
    Now, we have that $x_1 \ge x_2 \ge \cdots \ge x_n,$ and $x_i \ge y_i$ for all $i$. Now, let $r$ be the first index such that $x_r \le k$ (if $x_n > k,$ then set $r = n+1$). Now, if $r > p',$ then $x_i \ge k+1 \ge y_i'$ for all $i \in [p:p']$, and $x_i \ge y_i \ge y_i'$ for all $i \not\in [p:p']$, so we can do nothing to the sequence $\{x_i\}$ and we get $\{x_i'\}$ that dominates $\{y_i'\}$. Alternatively, we have that $p \le r \le p'$. By domination, we must have that $x_i = k$ for all $i \in [r:q]$. So, if we increase $x_i$ for all $i \in [r:p']$ by $1$ and decrease $x_i$ for all $i \in [s:q]$ by $1$, where $s$ is chosen so that $p'-r = q-s,$ then we claim $\{x_i'\}$ dominates $\{y_i'\}$. First, since $p'-p = q-q'$ and $p' \ge r \ge p,$ this means that $q \ge s \ge q',$ so $[r:p']$ and $[s:q]$ are disjoint nonempty sets. Next, note that $x_i = x_i' \ge y_i = y_i'$ for all $i \not\in [p:q]$. Next, for $p \le i < r,$ we have $x_i = x_i' \ge k+1 = y_i'$, and for $r \le i \le p',$ we have $x_i' = x_i+1 = k+1 = y_i'$. For $p' < i < s,$ we have $x_i' = k \ge y_i'$, and for $s \le i \le q,$ we have $x_i' = x_i-1 = k-1 = y_i'$. Therefore, $\{x_i'\}$ dominates $\{y_i'\}$.
\end{proof}

The final result in this section roughly tells us that to ensure majorization, a stone-variant-like time step is often better than anything the filler can hope to achieve in the standard variable-processor cup game.

\begin{lemma} \label{lem:toy_beats_standard}
    Suppose that $x_1 \ge x_2 \ge \cdots \ge x_n$ are all even integers, and let $\{x_i'\}$ be the fills after some filling and greedy emptying procedure. Then, there exists some integer $k \in \{x_1, \dots, x_n\}$ such that we can increase some (possibly $0$) of the $x_i$'s that equal $k$ to $k+2$, decrease an equal number of the $x_i$'s that equal $k$ to $k-2$, and the new sequence majorizes the sorted order of $\{x_i'\}$.
\end{lemma}

\begin{proof}
    Consider any filling procedure with fills $a_1, \dots, a_n$, followed by greedy emptying, on $\{x_i\}$. Let $k = x_p$.
    First, note that after filling and greedy emptying, $x_i' = x_i+a_i-1$ for $i \le p$ and $x_i' = x_i+a_i$ for $i > p$. Therefore, $|x_i'-x_i| \le 1$ for all $i$. In addition, if $x_i < k,$ then $i > p$ so $x_i' \ge x_i$, and if $x_i > k,$ then $i < p$ so $x_i' \le x_i$.
    
    Now, suppose that the set of $i$ such that $x_i = k$ is $[q:r] = \{q, q+1, \dots, r\}$ (note that $q \le p \le r$). Now, let $q' = \lfloor \frac{q+r-1}{2} \rfloor$ and $r' = \lceil \frac{q+r+1}{2}\rceil.$ Note that $[q:q']$ and $[r':r]$ have the same size. Also, if $[q:r]$ has even size, then $q \not\equiv r \bmod 2,$ so $r' = q'+1$, and if $[q:r]$ has odd size, then $r' = q'+2$.
    
    Now, let $x_i''$ be the sequence where $x_i'' = x_i$ if $x_i \not\equiv k$, and otherwise, $x_i'' = x_i+2$ for $q \le i \le q'$, $x_i'' = x_i-2$ for $r' \le i \le r,$ and $x_i'' = x_i$ for $q' < i < r'$. Note that $\{x_i''\}$ is sorted. In addition, since $\{x_i\}$ are all even integers, this means that even if $\{x_i'\}$ is not sorted, we have that $x_1', \dots, x_{q-1}' \ge x_q', \dots, x_r' \ge x_{r+1}', \dots, x_n'$.  Now, let $z_i$ be the sorted (in non-increasing order) of $x_i'$. Since $x_i \ge x_i'$ for all $1 \le i \le q-1,$ we have that $\{x_i\}_{i = 1}^{q-1}$ dominates $\{x_i'\}_{i = 1}^{q-1}$ by Proposition \ref{prop:super_dom_ez}, so $x_i \ge z_i$ for all $1 \le i \le q-1.$ Likewise, $\{x_i'\}_{r+1}^{n}$ dominates $\{x_i\}_{r+1}^{n}$, so $x_i \le z_i$ for all $r+1 \le i \le n$. Finally, we have that $z_i \in [k-1, k+1] = [x_i-1, x_i+1]$ for all $q \le i \le r.$
    
    We wish to show that $\{x_i''\}$ majorizes $\{z_i\}$. To do so, first we note that since $|[q:q']| = |[r':r]|,$ $\sum_{i = 1}^{n} x_i'' = \sum_{i = 1}^{n} x_i$, and we already know that $\sum_{i = 1}^{n} x_i = \sum_{i = 1}^{n} x_i'.$ Therefore, $\sum_{i = 1}^{n} x_i'' = \sum_{i = 1}^{n} z_i$. 
    
    It suffices to prove that for all $1 \le s \le n-1,$ we have $\sum_{i = 1}^{s} x_i'' \ge \sum_{i = 1}^{s} z_i$.
    If $s < q,$ then $$\sum_{i = 1}^{s} x_i'' = \sum_{i = 1}^{s} x_i \ge \sum_{i = 1}^{s} z_i.$$
    If $s \ge r,$ then $$\sum_{i = s+1}^{n} x_i'' = \sum_{i = s+1}^{n} x_i \le \sum_{i = s+1}^{n} z_i,$$ so $\sum_{i = 1}^{s} x_i'' \ge \sum_{i = 1}^{s} z_i$. 
    If $q \le s \le q',$ then $$\sum_{i = 1}^{s} x_i'' = \sum_{i = 1}^{q-1} x_i + \sum_{i = q}^{s} (x_i+2) \ge \sum_{i = 1}^{q-1} z_i + \sum_{i = q}^{s} z_i = \sum_{i = 1}^{s} z_i.$$ 
    If $r'-1 \le s \le r-1,$ then $$\sum_{i = s+1}^{n} x_i'' = \sum_{i = s+1}^{r} x_i'' + \sum_{i = r+1}^{n} x_i'' = \sum_{i = s+1}^{r} (x_i-2) +  \sum_{i = r+1}^{n} x_i \le \sum_{i = s+1}^{r} z_i \sum_{i = r+1}^{n} z_i = \sum_{i = s+1}^{n} z_i,$$ so $\sum_{i = 1}^{s} x_i'' \ge \sum_{i = 1}^{s} x_i'.$
    Since either $r'-1 = q'$ or $r'-1 = q'+1,$ we have covered all cases of $s$. Therefore, $\{x_i''\}$ majorizes $\{z_i\}$.
\end{proof}

\subsection{Finishing the Proof} \label{subsec:upper_finish}

In this section, we prove the main upper bound by combining our results on the $\ell$-variant cup game from Subsection \ref{subsec:l_variant} with our results on majorization and domination from Subsection \ref{subsec:majorization} that allow us to relate the optimal backlog in the $\ell$-variant cup game to the optimal backlog in the negative-fill variable-processor cup game. 

While it is well-known that the standard game is harder than the negative-fill game for the filler, it turns out that against a greedy emptier, the converse is also (asymptotically) true. Namely, we have the following lemma, the proof of which is deferred to Appendix \ref{sec:appendix_less_ez}:

\begin{lemma}
    For any $t \ge 1$, suppose the filler can guarantee that the backlog exceeds some $b \ge 0$ after $t$ rounds of the standard game, where the emptier is promised to be greedy. Then, the filler can guarantee that the backlog exceeds $b/2$ after $t$ rounds of the negative-fill game, where again the emptier is promised to be greedy.
\end{lemma}

Since the emptier is assumed to be greedy, from now on we may assume that the filler and emptier are competing in the negative-fill cup game. Thus, to prove Theorem \ref{thm:main_upper}, it suffices to instead prove the following theorem about the negative-fill game. 
\begin{theorem} \label{thm:main_rephrased}
    Let $1 \le k \le O(n).$ Then, assuming the emptier follows a greedy emptying strategy in the negative-fill variable-processor cup game, the filler needs $\Omega\left(k + \frac{k^3}{\log^2 (n/k)}\right)$ rounds to achieve backlog $k$.
\end{theorem}

Further note that the filler trivially needs at least $k$ rounds to achieve backlog $k$. Thus our task reduces to proving a lower bound of $\Omega\left(\frac{k^3}{\log^2 (n/k)}\right)$ on the number of rounds that the filler needs.

We now perform a series of arguments to establish Theorem \ref{thm:main_rephrased}.

\paragraph{Relating the stone-variant to the negative-fill game.}
Let $\{y^{(0)}\}, \{y^{(1)}\}, \{y^{(2)}\}$ represent sequences of fills where $\{y^{(i)}\}$ is the (sorted) sequence of fills after $i$ rounds of some negative-fill cup game.
In other words, the fills start out as $\{y^{(0)}\} = \{0, 0, \dots, 0\}.$ At time step $i+1$, the filler converts the sequence $\{y^{(i)}\} = \{y_1^{(i)}, \dots, y_n^{(i)}\}$, where $y_1^{(i)} \ge y_2^{(i)} \ge \cdots \ge y_n^{(i)},$ to $\{z^{(i)}\} = \{z_1^{(i)}, \dots, z_n^{(i)}\}$, where $z_1^{(i)} \ge z_2^{(i)} \ge \cdots \ge z_n^{(i)}$, by adding a total of $p_i$ fill, and finally, the emptier removes $1$ from each of the first $p_i$ cups and sorts the cups to obtain $\{y^{(i+1)}\}$. 
We claim that we can simulate a stone-variant cup game that starts with $\{x^{(0)}\} = \{0, 0, \dots, 0\},$ but for each $i$, if the cups after $i$ time steps have fills $\{x^{(i)}\} = \{x_1^{(i)}, \dots, x_n^{(i)}\},$ then $2 \cdot \{x^{(i)}\}$ majorizes $\{y^{(i)}\}$, and $\{x^{(i)}\}$ is entirely integer-valued. We prove this by induction on $i$. 

Clearly, for $i = 0,$ we have $2 \cdot \{x^{(i)}\} = (0, 0, \dots, 0) = \{y^{(i)}\}$, so $2 \cdot \{x^{(0)}\}$ majorizes $\{y^{(0)}\}$. Now, suppose that we have simulated the stone-variant cup game so that $2 \cdot \{x^{(i)}\}$ majorizes $\{y^{(i)}\}$ and $x_1^{(i)}, \dots, x_n^{(i)}$ are all integers. Then, by Lemma \ref{lem:majorization}, for any $\{y^{(i+1)}\}$ achievable from a round of the negative-fill cup game on $\{y^{(i)}\}$, there exists a sequence of fills on $2 \cdot \{x^{(i)}\}$ that, after greedy emptying, produces some $\{x'\}$ which majorizes $\{y^{(i+1)}\}$. Moreover, by Lemma \ref{lem:toy_beats_standard}, since all terms in $2 \cdot \{x^{(i)}\}$ are even integers, there exists a round of the stone-variant cup game (scaled by 2) applied to $2 \cdot \{x^{(i)}\}$ which majorizes $\{x'\}$. Therefore, there exists a round of the stone-variant cup game applied to $\{x^{(i)}\}$ to produce some $\{x^{(i+1)}\}$ such that $2 \cdot \{x^{(i+1)}\}$ majorizes $\{y^{(i+1)}\}.$ This completes the induction.

The above argument, along with the fact that $2 \cdot \{x^{(i+1)}\}$ majorizing $\{y^{(i+1)}\}$ implies that $2 \cdot x_1^{(i+1)} \ge y_1^{(i+1)}$, directly implies the following lemma.

\begin{lemma} \label{lem:N_var_finish}
    Suppose that after $T$ time steps of the negative-fill cup game against a greedy emptier, it is possible to have a fill of $k$ or more in some cup. Then, after $T$ time steps of the stone-variant cup game, it is possible to have a fill of $k/2$ or more in some cup.
\end{lemma}

\paragraph{Relating the $\ell$-variant to the stone-variant.}
Now, we compare the stone-variant cup game to the $\ell$-variant cup game for a suitable integer $\ell \ge 1.$ Specifically, for any sequence $\{x^{(i)}\}$ of the stone-variant cup game obtainable in $i$ time steps, and any $\ell \ge 1,$ we show there is a corresponding $\ell$-variant cup game that produces a sequence $\{w^{(i)}\}$ in $i$ time steps, where $\{w^{(i)}\}$ dominates $\{x^{(i)}\}$. Again, we prove this by induction, where the base case of $i = 0$ is clear since $\{w^{(i)}\} = (0, 0, \dots, 0) = \{x^{(i)}\}$. Now, if time step $i$ of the stone-variant cup game converts $\{x^{(i)}\}$ to $\{x^{(i+1)}\}$, then by Lemma \ref{lem:domination}, there exists a time step of the stone-variant cup game that converts $\{w^{(i)}\}$ to some $\{w'\}$ that dominates $\{x^{(i+1)}\}$. Therefore, if we did the same time step but in the $\ell$-variant cup game, each fill would either remain the same or go up by an additional $1$ to ensure that checkpoints are preserved, so we would reach a state $\{w^{(i+1)}\}$ that dominates $\{x^{(i+1)}\}$. This completes the induction, and along with Lemma \ref{lem:N_var_finish}, directly implies the following lemma.

\begin{lemma} \label{lem:ell_var_finish}
    Suppose that after $T$ time steps of the negative-fill cup game against a greedy emptier, it is possible to have a fill of $k$ or more in some cup. Then, after $T$ time steps of the $\ell$-variant cup game, it is possible to have a fill of $k/2$ or more in some cup. 
\end{lemma}

To finish the proof of Theorem \ref{thm:main_rephrased}, first suppose that $k \ge 16 \log n$. Let $k' = k/8$, $h = \lceil \log\left(\frac{n}{k'}+1\right) \rceil$, and $k''$ be the largest multiple of $h$ that is at most $k'$. Note that $k' \ge 2 \log n$ but $h \le \lceil \log n \rceil < k',$ so $k'' = \Theta(k)$. Also, we have that $h = \Theta\left(\log \left(\frac{n}{k''}+1\right)\right)$, that $h \ge 2$, and that $h|k''$. Set $\ell = k'' / h$. By Lemma \ref{lem:ell_var_finish}, if the filler in the negative-fill cup game can achieve a fill of $k$ after $T$ time steps (against a greedy emptier), then it is also possible to achieve a fill of $k/2 \ge 4k''$ after $T$ time steps in the $\ell$-variant cup game. We can therefore apply Theorem \ref{thm:inequality} to deduce that $T \ge \Omega((k'')^3/h^2) = \Omega\left(\frac{k^3}{\log (n/k)^2}\right)$. Finally, since the backlog cannot increase by more than $1$ in a step, $T \ge k$, even if $k < 16 \log n$. This completes the proof of Theorem \ref{thm:main_rephrased}.

\section{Analyzing Greedy with Resource Augmentation} \label{sec:resource_augmentation}
In this section, we analyze the variable-processor cup game with $\eps $ resource augmentation, meaning that in each step, the emptier is permitted to remove up to $ 1+\eps $ units of water from each of $ p $ cups (rather than just $ 1 $ unit of water). Of course, so that the resource augmentation does not cause the average fill of cups to become negative, the emptier is prohibited from reducing the fill of any cup below $0$.

We prove that even a very small amount of resource augmentation significantly decreases backlog of the game:

\begin{theorem}
In the variable-processor cup game with resource augmentation $\eps \in [\omega(1/n), 1] $, the greedy emptying algorithm achieves backlog
$ O (\eps ^ { -1 }\log n) $.
\label{thm:augmentation}
\end{theorem}

We prove the theorem using the probabilistic method. Rather than analyzing the greedy emptying algorithm directly, we instead construct a randomized emptying algorithm that, at any given moment, achieves backlog $ O (\eps ^ { -1 }\log n) $ with non-zero probability. (Importantly, the randomized algorithm is against an \emph{adaptive} filler, not an oblivious one.) The existence of such a randomized algorithm implies the existence of a deterministic emptying algorithm with the same guarantee. But by Corollary \ref{cor:aug}, we already know that the greedy emptying algorithm is the optimal deterministic emptying algorithm, and thus it must also achieve backlog $ O (\eps ^ { -1 }\log n) $.

Suppose $\omega(1 /n) \le \eps \le 1/2$. Throughout the rest of the section, we make an important WLOG assumption about the filler's behavior: we will assume that the filler always places at least $1/n^2$ units of water into each cup on each step. To see why this assumption is WLOG, up to constant factor changes in $\eps$ and in the backlog, let us describe how we can obtain the assumption as the composition of two more obviously WLOG assumptions. We can assume WLOG that the filler is only able to place $(1 - \eps)p$ units of water into cups in a given step (rather than $p$ units of water), with at most $(1 - \eps)$ units of water going to any given cup (rather than $1$ unit of water). In particular, this is equivalent to re-normalizing what $ 1 $ unit of water corresponds to by a factor of $(1 - \eps)$, and then changing $\eps$ by a constant factor. We can then further modify the filler to place an additional $\eps p / n$ units of water into each cup (on top of the water that the filler was originally going to place); this modification is also WLOG since it only helps the filler. The new filler places at most $p$ units of water into the system at each step, with at most $1$ unit of water going to any given cup, and with at least $\eps p / n \ge 1 / n^2$ units of water going to each cup. 

The assumption that the filler always places at least $1/n^2$ units of water into each cup on each step makes it so that a very simple randomized emptying algorithm achieves small backlog, even in arbitrarily long games. The randomized emptying algorithm takes an approach that we call  \defnn{proportional emptying}: in each step of the game, if the filler places some amount $ q_j $ of water into cup $ j $, then the emptier empties from cup $ j $ with probability exactly $ q_j $. (If the algorithm empties from a cup with less than $1 + \eps$ water, then the amount of water in the cup becomes $0$).

Note that, if we did not have our WLOG assumption about the filler, then one problem with proportional emptying would be that, in very long games, the filler could use the following strategy to achieve large backlog. The filler simply waits until they get lucky and have a cup $j$ with some large backlog, and then the filler never places any more water into that cup $ j $. This means that the emptier will never again empty from cup $ j $, and the backlog will stay large forever. This is why, in the analysis of proportional emptying, we need the assumption that the filler always places at least $1/n^2$ units of water into each cup on each step; of course, if the assumption weren't there, we could just modify proportional emptying to simulate the assumption (since the assumption is WLOG), but it is cleaner to just assume it.

The following lemma shows that it is possible to implement proportional emptying on any given step.

\begin{lemma}\label{lem:proportion}
Let $p \in \mathbb{N}$ and let $q_1, q_2, \ldots, q_n \in [0, 1]$ satisfy $\sum_j q_j = p$. It is possible to select distinct $x_1, x_2, \ldots, x_p \in [n]$ in a random fashion such that for each $j \in [n]$,
$$\Pr[j \in \{x_i\}] = q_j.$$
\end{lemma}
\begin{proof}
Partition the real interval $ I = [0, p) $ into half-open intervals $ Q_1, Q_2,\ldots, Q_n $ where $| Q_j| = q_j $ for all $ j $. Select $u_1 \in [0, p)$ uniformly at random, and define $ u_2 = u_1+1, u_3 = u_1+2,\ldots, u_n = u_1+ n - 1$, where each $u_i$ is taken modulo $p$. Finally, define $ x_i $ to be the index $ j $ such that $ u_i\in Q_j $.

Since each interval $Q_j$ has size $1$ or smaller, it can contain at most one $ u_i $. For each $ i \in [n]$, the value $u_i$ is uniformly random in $[0, p)$, and thus $\Pr [u_i \in Q_j] = q_j / p $. Since the events $\{u_i \in Q_j\}_{i = 1}^p$ are disjoint, it follows that 
$$\Pr[u_i \in Q_j \text{ for any }i \in [p]] = \sum_{i = 1 }^{p}\Pr [u_i \in Q_j] = q_j.$$
Thus $\Pr[j \in \{x_i\}] = q_j$ and the lemma is proven.
\end{proof}

To analyze proportional emptying, we show that, at any given moment, each cup has fill $ O (\eps ^ { -1 }\log n) $ with high probability. Roughly speaking, the amount of water in each cup can then be modeled as a biased random walk: in each step, the expected amount of water that the emptier removes from the cup is a factor of $ 1+\eps $ larger than the amount of water that the filler inserts. The bias in the random walk prevents it from reaching a large fill at any given moment, and as we shall now see, a simple Chernoff-style analysis (modified to handle the fact that the filler is an adaptive adversary) can be used to bound the fill by $ O (\eps ^ { -1 }\log n) $.

In our analysis of proportional emptying, we will need the following standard Chernoff-style bound, which can be viewed as a variant of Azuma's martingale inequality:
\begin{proposition}\label{prop:azuma}
Fix $\mu, t \in \mathbb{N}$. Suppose that Alice constructs a sequence of random variables $X_1, \ldots, X_t$, with $X_i \in \{0, 1\}$, using the following iterative process. Once the outcomes of $X_1, \ldots, X_{i - 1}$ are determined, Alice then selects a probability $r_i$ and draws $X_i$ from $\{0, 1\}$ with $\Pr[X_i = 1] = r_i$. Alice is an adaptive adversary, meaning that $r_i$ can be a function of $X_1, \ldots, X_{i - 1}$. The only constraint on Alice is that $\sum_i r_i \ge \mu$.

If $X = \sum_i X_i$, then for any $j > 0$,
\[\Pr[X\leq \mu - j] \leq \exp\left(j^2  / (2\mu) \right).\]
\end{proposition}

If Alice were an oblivious adversary, unable to see the outcomes of the $X_i$s, then Proposition \ref{prop:azuma} would follow by a standard lower-tail Chernoff bound. Fortunately, Chernoff bounds still hold against an adaptive adversary like Alice. For a proof of Proposition \ref{prop:azuma}, see \cite{azuma}; Corollary 11 of \cite{azuma} gives the upper-tail version of Proposition \ref{prop:azuma} and by modifying the argument in the standard way to get a lower-tail bound, one arrives at Proposition \ref{prop:azuma}.

We complete the section by analyzing proportional emptying.
\begin{lemma}
Suppose $\eps \in (\omega(1/n), 1]$. If the emptier uses proportional emptying, and if the filler always places at least $1/n^2$ units of water into each cup on each step, then at any time $ t $, the backlog is $ O (\eps ^ { -1 }\log n) $ with high probability in $ n $.
\label{lem:prop}
\end{lemma}
\begin{proof}
We will show that, for any given cup $ j $ and time $ t_0 $, with high probability in $ n $, the amount of water in cup $ j $ is at most $ O (\eps ^ { -1 }\log n) $. By taking a union bound over $ j\in [n] $, it follows that the backlog at time $ t_0 $ is $ O (\eps ^ { -1 }\log n) $ with high probability in $ n $.

Since the cups are defined symmetrically, let us just analyze the amount of water in cup $ 1 $ over time. Let $ c_i $ be the amount of water that the filler places into cup $ 1 $
in step $i$. Recall that the filler is an adaptive adversary, so $ c_i $ may be determined partially based on the emptier's behavior in the first $ i - 1$ steps.

Let $r$ be the largest $r \le t_0$ such that cup $1$ is empty at the end of step $r$, and define $t = t_0-r$. Let $\mu$ be the amount of water that the filler has placed into cup $1$ during steps $r+1 , \ldots, t_0$, and let $X$ be the number of times that the emptier has emptied from cup $1$ during those same steps. The amount of water in cup $ 1$ after step $t$ is
exactly $$\mu - (1 + \eps)X.$$
To complete the analysis, it suffices to show that
\begin{equation}
\Pr\left[X \le \frac{1}{1 + \eps} \cdot \lfloor \mu \rfloor - \gamma \eps^{-1}\log n\right] \le \frac{1}{\poly(n)},
\label{eq:walk1}
\end{equation}
where $\gamma$ is taken to be a sufficiently large constant. 

For any integer $\mu \ge 1$ and any $1 \le t \le t_0$, if the end of round $r = t_0-t$ is the last time that the fill of cup $1$ is $0$, and the integer part (floor) of the amount of water placed into cup 1 during steps $r + 1, \ldots, t_0$ equals $\mu$, define $X_{\mu, t}$ to be the number of times that the emptier has removed water from cup 1 during steps $t_0, t + 1, \ldots, t$. Otherwise, define $X_{\mu, t}$ to be $\infty$. To bound \eqref{eq:walk1}, it suffices to bound the probability that there exists any $t$ and any integer $\mu \ge 1$ such that $X_{\mu, t} \le  \frac{1}{1 + \eps} \mu - \gamma \eps^{-1}\log n$. For any given $\mu, t$, we have by Proposition \ref{prop:azuma} that
\begin{align*}
\Pr\left[X_{\mu, t} \le \frac{1}{1 + \eps} \mu - \gamma \eps^{-1}\log n\right] \le \exp\left(- \left(\frac{\eps \mu + \gamma \eps^{-1}\log n}{\sqrt{2\mu}}\right)^2\right).
\end{align*}
In addition to assuming that $\mu \ge 1$ and $t$ are integers, we can also assume that $\mu \ge t / n^2$ since, by assumption, each step places at least $1/n^2$ water into cup $1$.

Taking a union bound over all possible $\mu \ge 1$ and $t \ge 0$ satisfying $\mu \ge t / n^2$, we have
\begin{align*}
&\hspace{0.5cm} \Pr[X \le \frac{1}{1 + \eps} \mu - \gamma \eps^{-1}\log n] \\
& \le \sum_{t = 1}^{n} \sum_{\mu = \lceil t/n^2\rceil}^{tn} \Pr[X_{\mu, t} \le  \frac{1}{1 + \eps} \mu - \gamma \eps^{-1}\log n] \\
& \le \sum_{\mu = 1}^\infty \sum_{t = 1}^{n^2\mu} \Pr[X_{\mu, t} \le  \frac{1}{1 + \eps} \mu - \gamma \eps^{-1}\log n] \\
& = \sum_{\mu = 1}^\infty n^2 \mu \Pr[X_{\mu, t} \le  \frac{1}{1 + \eps} \mu - \gamma \eps^{-1}\log n] \\
& \le \sum_{\mu = 1}^\infty n^2 \mu \exp\left(- \left(\frac{\eps \mu + \gamma \eps^{-1}\log n}{\sqrt{2\mu}}\right)^2\right) \\
& \le \sum_{\mu = 1}^{\gamma \eps^{-2}\log n} n^2 \mu \exp\left(- \left(\frac{\gamma \eps^{-1}\log n}{\sqrt{2\mu}}\right)^2\right) + \sum_{\mu > \gamma \eps^{-2}\log n} n^2 \mu \exp\left(- \left(\frac{\eps \mu}{\sqrt{2\mu}}\right)^2\right) \\
& \le \sum_{\mu = 1}^{\gamma \eps^{-2}\log n} n^2\mu \exp\left(- \frac{\gamma \log n}{2}\right) + \sum_{\mu > \gamma \eps^{-2}\log n} n^2 \mu \exp\left(- \eps^2 \mu/2\right) \\
& \le \frac{1}{\poly(n)},
\end{align*}
where the final inequality uses that $\gamma$ is a sufficiently large constant. 
\end{proof}

\begin{proof}[Proof of Theorem \ref{thm:augmentation}]
Let $c$ be a sufficiently large constant. Consider the two-player game of length $t$ where the filler is declared to win if he achieves backlog at least $c \eps ^ { -1 }\log n$ and the emptier wins otherwise. This is a perfect-information game with no draws, so one of the players must be able to deterministically force a win. Lemma \ref{lem:prop} (along with the fact that the filler WLOG places at least $1/n^2$ water into each cup on each step) tells us that, if $c$ is sufficiently large, then the filler cannot always force a win. Therefore the emptier \emph{can} always force a win, and there is a deterministic emptying algorithm that achieves backlog $c \eps^{-1} \log n$. By Corollary \ref{cor:aug}, it follows that greedy emptying also achieves backlog $c\eps^{-1} \log n$, completing the proof. 
\end{proof}

Interestingly, the argument above can also be used to obtain a nontrivial bound in the resource-augmentation-free setting. Now the amount of water in each cup follows an \emph{unbiased} random walk. At any given step $t$, one can bound the height of such random walk by $O(\sqrt{t\log n}))$ with high probability. Using the fact that greedy emptying is as good as any randomized emptying strategy, it follows that greedy emptying achieves backlog $O(\sqrt{t\log n})$ in a $t$ step game.

\section*{Acknowledgments}

William Kuszmaul is funded by a Fannie \& John Hertz Foundation Fellowship; by an NSF GRFP Fellowship; and by the United States Air Force Research Laboratory and the United States Air Force Artificial Intelligence Accelerator and was accomplished under Cooperative Agreement Number FA8750-19-2-1000. The views and conclusions contained in this document are those of the authors and should not be interpreted as representing the official policies, either expressed or implied, of the United States Air Force or the U.S. Government. The U.S. Government is authorized to reproduce and distribute reprints for Government purposes notwithstanding any copyright notation herein.the United States Air Force Research Laboratory under Cooperative Agreement Number FA8750-19-2-1000. The views and conclusions contained in this document are those of the authors and should not be interpreted as representing the official policies, either expressed or implied, of the United States Air Force or the U.S. Government. The U.S. Government is authorized to reproduce and distribute reprints for Government purposes notwithstanding any copyright notation herein.

Shyam Narayanan is funded by an NSF GRFP Fellowship, a Simons Investigator Award, and the NSF TRIPODS Program (award DMS-2022448).


\begin{thebibliography}{10}

\bibitem{AdlerBeFr03}
Micah Adler, Petra Berenbrink, Tom Friedetzky, Leslie~Ann Goldberg, Paul
  Goldberg, and Mike Paterson.
\newblock A proportionate fair scheduling rule with good worst-case
  performance.
\newblock In {\em Proceedings of the Fifteenth Annual ACM Symposium on Parallel
  Algorithms and Architectures (SPAA)}, pages 101--108, 2003.
\newblock URL: \url{http://doi.acm.org/10.1145/777412.777430}, \href
  {https://doi.org/10.1145/777412.777430} {\path{doi:10.1145/777412.777430}}.

\bibitem{AmirFaId95}
Amihood Amir, Martin Farach, Ramana~M. Idury, Johannes A.~La Poutr{\'{e}}, and
  Alejandro~A. Sch{\"{a}}ffer.
\newblock Improved dynamic dictionary matching.
\newblock {\em Inf. Comput.}, 119(2):258--282, 1995.
\newblock \href {https://doi.org/10.1006/inco.1995.1090}
  {\path{doi:10.1006/inco.1995.1090}}.

\bibitem{AmirFr14}
Amihood Amir, Gianni Franceschini, Roberto Grossi, Tsvi Kopelowitz, Moshe
  Lewenstein, and Noa Lewenstein.
\newblock Managing unbounded-length keys in comparison-driven data structures
  with applications to online indexing.
\newblock {\em SIAM Journal on Computing}, 43(4):1396--1416, 2014.

\bibitem{AzarLi06}
Yossi Azar and Arik Litichevskey.
\newblock Maximizing throughput in multi-queue switches.
\newblock {\em Algorithmica}, 45(1):69--90, 2006.

\bibitem{Bar-NoyFrLa02}
Amotz Bar-Noy, Ari Freund, Shimon Landa, and Joseph~(Seffi) Naor.
\newblock Competitive on-line switching policies.
\newblock In {\em Proceedings of the Thirteenth Annual ACM-SIAM Symposium on
  Discrete Algorithms (SODA)}, pages 525--534, 2002.
\newblock URL: \url{http://dl.acm.org/citation.cfm?id=545381.545452}.

\bibitem{BarNi02}
Amotz Bar-Noy, Aviv Nisgav, and Boaz Patt-Shamir.
\newblock Nearly optimal perfectly periodic schedules.
\newblock {\em Distributed Computing}, 15(4):207--220, 2002.

\bibitem{BaruahCoPl96}
S.~K. Baruah, N.~K. Cohen, C.~G. Plaxton, and D.~A. Varvel.
\newblock Proportionate progress: A notion of fairness in resource allocation.
\newblock {\em Algorithmica}, 15(6):600--625, Jun 1996.
\newblock \href {https://doi.org/10.1007/BF01940883}
  {\path{doi:10.1007/BF01940883}}.

\bibitem{BaruahGe95}
Sanjoy~K Baruah, Johannes~E Gehrke, and C~Greg Plaxton.
\newblock Fast scheduling of periodic tasks on multiple resources.
\newblock In {\em Proceedings of the 9th International Parallel Processing
  Symposium}, pages 280--288, 1995.

\bibitem{BenderDaFa20}
Michael Bender, Rathish Das, Mart\'in Farach-Colton, Rob Johnson, and William
  Kuszmaul.
\newblock Flushing without cascades.
\newblock In {\em Proceedings of the Thirty-First Annual ACM-SIAM Symposium on
  Discrete Algorithms (SODA)}, 2020.

\bibitem{BenderFaKu19}
Michael Bender, Mart\'in Farach-Colton, and William Kuszmaul.
\newblock Achieving optimal backlog in multi-processor cup games.
\newblock In {\em Proceedings of the 51st Annual ACM Symposium on Theory of
  Computing (STOC)}, 2019.

\bibitem{CA2}
Michael~A Bender, Rezaul~A Chowdhury, Rathish Das, Rob Johnson, William
  Kuszmaul, Andrea Lincoln, Quanquan~C Liu, Jayson Lynch, and Helen Xu.
\newblock Closing the gap between cache-oblivious and cache-adaptive analysis.
\newblock In {\em Proceedings of the 32nd ACM Symposium on Parallelism in
  Algorithms and Architectures (SPAA)}, pages 63--73, 2020.

\bibitem{CA3}
Michael~A Bender, Rezaul~A Chowdhury, Rathish Das, Rob Johnson, William
  Kuszmaul, Andrea Lincoln, Quanquan~C Liu, Jayson Lynch, and Helen Xu.
\newblock Closing the gap between cache-oblivious and cache-adaptive analysis.
\newblock In {\em Proceedings of the 32nd ACM Symposium on Parallelism in
  Algorithms and Architectures (SPAA)}, pages 63--73, 2020.

\bibitem{CA4}
Michael~A Bender, Roozbeh Ebrahimi, Jeremy~T Fineman, Golnaz Ghasemiesfeh, Rob
  Johnson, and Samuel McCauley.
\newblock Cache-adaptive algorithms.
\newblock In {\em Proceedings of the twenty-fifth annual ACM-SIAM Symposium on
  Discrete Algorithms (SODA)}, pages 958--971. SIAM, 2014.

\bibitem{BenderKu20}
Michael~A Bender and William Kuszmaul.
\newblock Randomized cup game algorithms against strong adversaries.
\newblock In {\em Proceedings of the 2021 ACM-SIAM Symposium on Discrete
  Algorithms (SODA)}, pages 2059--2077. SIAM, 2021.

\bibitem{DamaschkeZh05}
Peter Damaschke and Zhen Zhou.
\newblock On queuing lengths in on-line switching.
\newblock {\em Theoretical computer science}, 339(2-3):333--343, 2005.

\bibitem{DietzSl87}
Paul Dietz and Daniel Sleator.
\newblock Two algorithms for maintaining order in a list.
\newblock In {\em Proceedings of the Nineteenth Annual ACM Symposium on Theory
  of Computing (STOC)}, pages 365--372, 1987.
\newblock URL: \url{http://doi.acm.org/10.1145/28395.28434}, \href
  {https://doi.org/10.1145/28395.28434} {\path{doi:10.1145/28395.28434}}.

\bibitem{DietzRa91}
Paul~F. Dietz and Rajeev Raman.
\newblock Persistence, amortization and randomization.
\newblock In {\em Proceedings of the Second Annual ACM-SIAM Symposium on
  Discrete Algorithms (SODA)}, pages 78--88, 1991.
\newblock URL: \url{http://dl.acm.org/citation.cfm?id=127787.127809}.

\bibitem{FischerGa15}
Johannes Fischer and Pawe{\l} Gawrychowski.
\newblock Alphabet-dependent string searching with wexponential search trees.
\newblock In {\em Annual Symposium on Combinatorial Pattern Matching (CPM)},
  pages 160--171, 2015.

\bibitem{FleischerKo04}
Rudolf Fleischer and Hisashi Koga.
\newblock Balanced scheduling toward loss-free packet queuing and delay
  fairness.
\newblock {\em Algorithmica}, 38(2):363--376, Feb 2004.
\newblock \href {https://doi.org/10.1007/s00453-003-1064-z}
  {\path{doi:10.1007/s00453-003-1064-z}}.

\bibitem{Gail93}
H~Richard Gail, G~Grover, Roch Gu{\'e}rin, Sidney~L Hantler, Zvi Rosberg, and
  Moshe Sidi.
\newblock Buffer size requirements under longest queue first.
\newblock {\em Performance Evaluation}, 18(2):133--140, 1993.

\bibitem{GkasieniecKl17}
Leszek Gasieniec, Ralf Klasing, Christos Levcopoulos, Andrzej Lingas, Jie Min,
  and Tomasz Radzik.
\newblock Bamboo garden trimming problem (perpetual maintenance of machines
  with different attendance urgency factors).
\newblock In {\em International Conference on Current Trends in Theory and
  Practice of Informatics}, pages 229--240. Springer, 2017.

\bibitem{Goldwasser10}
Michael~H. Goldwasser.
\newblock A survey of buffer management policies for packet switches.
\newblock {\em {SIGACT} News}, 41(1):100--128, 2010.
\newblock URL: \url{http://doi.acm.org/10.1145/1753171.1753195}, \href
  {https://doi.org/10.1145/1753171.1753195}
  {\path{doi:10.1145/1753171.1753195}}.

\bibitem{GoodrichPa13}
Michael~T Goodrich and Pawe{\l} Pszona.
\newblock Streamed graph drawing and the file maintenance problem.
\newblock In {\em International Symposium on Graph Drawing}, pages 256--267.
  Springer, 2013.

\bibitem{GuanYi12}
Nan Guan and Wang Yi.
\newblock Fixed-priority multiprocessor scheduling: Critical instant, response
  time and utilization bound.
\newblock In {\em 2012 IEEE 26th International Parallel and Distributed
  Processing Symposium Workshops \& PhD Forum}, pages 2470--2473. IEEE, 2012.

\bibitem{im2021matroid}
Sungjin Im, Benjamin Moseley, and Rudy Zhou.
\newblock The matroid cup game.
\newblock {\em Operations Research Letters}, 49(3):405--411, 2021.

\bibitem{Kopelowitz12}
Tsvi Kopelowitz.
\newblock On-line indexing for general alphabets via predecessor queries on
  subsets of an ordered list.
\newblock In {\em Proceedings of the 53rd Annual Symposium on Foundations of
  Computer Science (FOCS)}, pages 283--292, 2012.

\bibitem{Kuszmaul20}
William Kuszmaul.
\newblock Achieving optimal backlog in the vanilla multi-processor cup game.
\newblock In {\em Proceedings of the Thirty-First Annual ACM-SIAM Symposium on
  Discrete Algorithms (SODA)}, 2020.

\bibitem{tail}
William Kuszmaul.
\newblock How asymmetry helps buffer management: achieving optimal tail size in
  cup games.
\newblock In {\em Proceedings of the 53rd Annual ACM SIGACT Symposium on Theory
  of Computing (STOC)}, pages 1248--1261, 2021.

\bibitem{azuma}
William Kuszmaul and Qi~Qi.
\newblock The multiplicative version of azuma's inequality, with an application
  to contention analysis.
\newblock {\em arXiv preprint arXiv:2102.05077}, 2021.

\bibitem{variable}
William Kuszmaul and Alek Westover.
\newblock The variable-processor cup game.
\newblock In {\em 12th Innovations in Theoretical Computer Science Conference
  (ITCS)}. Schloss Dagstuhl-Leibniz-Zentrum f{\"u}r Informatik, 2021.

\bibitem{CA5}
Andrea Lincoln, Quanquan~C Liu, Jayson Lynch, and Helen Xu.
\newblock Cache-adaptive exploration: Experimental results and scan-hiding for
  adaptivity.
\newblock In {\em Proceedings of the 30th on Symposium on Parallelism in
  Algorithms and Architectures (SPAA)}, pages 213--222, 2018.

\bibitem{LitmanMo05}
Ami Litman and Shiri Moran-Schein.
\newblock On distributed smooth scheduling.
\newblock In {\em Proceedings of the Seventeenth Annual ACM Symposium on
  Parallelism in Algorithms and Architectures (SPAA)}, pages 76--85, 2005.

\bibitem{LitmanMo09}
Ami Litman and Shiri Moran-Schein.
\newblock Smooth scheduling under variable rates or the analog-digital
  confinement game.
\newblock {\em Theor. Comp. Sys.}, 45(2):325--354, June 2009.
\newblock URL: \url{http://dx.doi.org/10.1007/s00224-008-9134-x}, \href
  {https://doi.org/10.1007/s00224-008-9134-x}
  {\path{doi:10.1007/s00224-008-9134-x}}.

\bibitem{LitmanMo11}
Ami Litman and Shiri Moran-Schein.
\newblock On centralized smooth scheduling.
\newblock {\em Algorithmica}, 60(2):464--480, 2011.

\bibitem{Liu69}
Chung~Laung Liu.
\newblock Scheduling algorithms for multiprocessors in a hard real-time
  environment.
\newblock {\em JPL Space Programs Summary, 1969}, 1969.

\bibitem{LiuLa73}
Chung~Laung Liu and James~W Layland.
\newblock Scheduling algorithms for multiprogramming in a hard-real-time
  environment.
\newblock {\em Journal of the ACM (JACM)}, 20(1):46--61, 1973.

\bibitem{CA1}
Richard~T Mills, Andreas Stathopoulos, and Dimitrios~S Nikolopoulos.
\newblock Adapting to memory pressure from within scientific applications on
  multiprogrammed cows.
\newblock In {\em Proc. 8th International Parallel and Distributed Processing
  Symposium (IPDPS)}, page~71, 2004.

\bibitem{MoirRa99}
Mark Moir and Srikanth Ramamurthy.
\newblock Pfair scheduling of fixed and migrating periodic tasks on multiple
  resources.
\newblock In {\em Proceedings of the 20th IEEE Real-Time Systems Symposium},
  pages 294--303, 1999.

\bibitem{Mortensen03}
Christian~Worm Mortensen.
\newblock Fully-dynamic two dimensional orthogonal range and line segment
  intersection reporting in logarithmic time.
\newblock In {\em Proceedings of the Fourteenth Annual ACM-SIAM Symposium on
  Discrete Algorithms (SODA)}, pages 618--627, 2003.

\bibitem{RosenblumGoTa04}
Michael Rosenblum, Michel~X Goemans, and Vahid Tarokh.
\newblock Universal bounds on buffer size for packetizing fluid policies in
  input queued, crossbar switches.
\newblock In {\em Twenty-third Annual Joint Conference of the IEEE Computer and
  Communications Societies (INFOCOM)}, volume~2, pages 1126--1134, 2004.

\end{thebibliography}

\appendix


\section{Asymptotic Equivalence between Negative-Fill and Standard Game} \label{sec:appendix_less_ez}

We recall that in the negative-fill cup game, we allow cups to have arbitrary real number (possibly negative) fills, whereas in the standard variable-processor cup game, if an emptier removes water from a cup of fill between $0$ and $1$, the fill goes down to $0$ rather than a negative number.
In both games, if the cups have fills $x_1, \dots, x_n$ the goal of the filler is to maximize the backlog, where the backlog is $\max_{1 \le i \le n} x_i$ as opposed to $\max_{1 \le i \le n} |x_i|$, and the goal of the emptier is to minimize the backlog.

It is easy to see that the negative-fill cup game is harder for the filler than in the standard variable-processor cup game. That is, if the emptier can guarantee the backlog does not exceed some $b \ge 0$ after $t$ rounds of the standard variable-processor game, then they can make the same guarantee in the negative-fill game (by simply simulating it to be as though they never empty any cup below $0$).

In this section, we show the converse is also true, up to a factor of $2$ for backlog. First, we show the converse is true if we are promised the emptier is greedy. We note that this is all that we need for our main theorems, though we will also show a corollary that removes the greedy-emptier assumption.

First, we establish the following proposition.

\begin{proposition} \label{prop:binks:}
    For any integers $n, t \ge 1$ and real number $b \ge 0$, suppose the filler is playing the negative-fill variable-processor cup game over $n$ cups against a greedy emptier, for $t$ rounds. Then, the following are equivalent:
\begin{enumerate}
    \item The filler can guarantee at least one cup has fill $\ge b$.
    \item The filler can guarantee at least one cup has fill $\le -b$.
    \item The filler can either guarantee that at least one cup has fill $\ge b$ or at least one cup has fill $\le -b$ (but does not have to know which one).
\end{enumerate}
\end{proposition}

\begin{proof}
    We first establish the equivalence of the first two items. Since the emptier's strategy is fixed, the filler can know exactly what the state is after each round. Suppose that the filler can play a strategy where after the $i$th round, the state of the cups is $X^{(i)} = (x_1^{(i)}, \dots, x_n^{(i)})$. Note that $X^{(0)} = (0, 0, \dots, 0)$. Then, we claim that the filler is also capable of ensuring that after the $i$th round, the state of the cups is $-X^{(i)} = (-x_1^{(i)}, \dots, -x_n^{(i)})$.
    
    To do so, suppose the filler can go from $X^{(i)}$ to $X^{(i+1)}$ by adding $a_j^{(i)}$ to the $j$th cup and then having the emptier remove $1$ unit from the $p = \sum_{j = 1}^{n} a_j^{(i)}$ fullest cups. Then, to go from $-X^{(i)}$ to $-X^{(i+1)}$, the filler will add $1-a_j^{(i)}$ to the $j$th cup, so that the $j$th cup has fill $1-(X_j^{(i)}+a_j^{(i)})$. Then, $\sum_{j = 1}^{n} (1-a_j^{(i)}) = n-p$, so the emptier will subtract $1$ from the $n-p$ largest values of $1-(X_j^{(i)}+a_j^{(i)})$. This is equivalent to subtracting $1$ from all of the values $1-(X_j^{(i)}+a_j^{(i)})$ to get $-(X_j^{(i)}+a_j^{(i)})$, and then adding $1$ to the $p$ smallest values of $-(X_j^{(i)}+a_j^{(i)})$. Overall, we will precisely get $-X_j^{(i+1)}$, since the $p$ smallest values of $-(X_j^{(i)}+a_j^{(i)})$ are just the negatives of the $p$ largest values of $X_j^{(i)}+a_j^{(i)}$. Since we only care about cups up to permutations, we do not have to worry about tiebreaking issues.
    
    This establishes that the first two items are equivalent, since if after $t$ rounds the filler can guarantee the fills are $(X_1^{(t)}, \dots, X_n^{(t)})$ with $\max_{1 \le j \le n} X_j^{(t)} \ge b$, then the filler can guarantee the fills are $(-X_1^{(t)}, \dots, -X_n^{(t)})$, where $\min_{1 \le j \le n} -X_j^{(t)} \le -b$, and vice versa.
    
    Finally, to see why the third item is equivalent, it is straightforward that either the first two items implies the third item. But since the filler is really playing a $1$-player game as the emptier's strategy is known beforehand, if the filler can guarantee one of the two events can occur (but doesn't know which one), then the filler can either guarantee the first event will occur or can guarantee that the second event occurs. But since these two are equivalent, this establishes the equivalence between all three items.
\end{proof}

Now, we are able to establish the main lemma of this section, assuming the emptier is greedy.

\begin{lemma} \label{lem:neg_nonneg_equiv}
    For any $n, t \ge 1$, suppose the filler can guarantee that the backlog exceeds some $b \ge 0$ after $t$ rounds of the standard variable-processor game over $n$ cups, where the emptier is promised to be greedy. Then, the filler can guarantee that the backlog exceeds $b/2$ after $t$ rounds of the negative-fill game, where again the emptier is promised to be greedy.
\end{lemma}

\begin{proof}
    Since adding fill always benefits the filler, we know that if we were to start in a state where all cups have fills $b/2 > 0$, then the filler would still be able to guarantee that the backlog exceeds $b$ after $t$ rounds in the standard variable-processor game against a greedy emptier. In order for a cup to reach fill $b$ starting from this state, either some cup must reach $0$ fill at or before time step $t$, or some cup must reach fill $b$ without any cup reaching $0$ fill during the $t$ time steps. 
    
    The former case means that the same filling procedure on the negative-fill game starting with all cups with fill $0$ will result in some cup having fill less than or equal to $-b/2$ at some point during the first $t$ time steps. If this takes fewer than $t$ time steps, then the filler can preserve the state of the cups for the remaining steps (by always setting $p = n$ and putting $1$ unit of water into every cup on each step) so that  after $t$ time steps, the minimum fill is still less than or equal to $-b/2$. The latter case, on the other hand, means that in the negative-fill cup game starting with all cups with fill $0$ will result in some cup having fill $b/2$ at $t$ time steps. Therefore, the filler can guarantee that either the minimum fill is $-b/2$ after $t$ rounds, or the maximum fill is $b/2$ after $t$ rounds. Thus, by Proposition \ref{prop:binks:}, the filler can guarantee a backlog of $b/2$ after $t$ rounds in the negative-fill cup game.
\end{proof}

By noting the fact that the greedy emptier is optimal in both the negative-fill and standard variable-processor cup games, we have the following corollary.

\begin{corollary} \label{cor:neg_nonneg_equiv}
    For any $n, t \ge 1$, suppose the filler can guarantee that the backlog exceeds some $b \ge 0$ after $t$ rounds of the standard variable-processor game over $n$ cups, against any emptier. Then, the filler can guarantee that the backlog exceeds $b/2$ after $t$ rounds of the negative-fill game, against any emptier.
\end{corollary}

\section{Propositions on Majorization and Domination} \label{sec:appendix_ez}

In this section, we prove Propositions \ref{prop:ez2}, \ref{prop:perturbation}, and \ref{prop:super_dom_ez}.

\begin{proof}[Proof of Proposition \ref{prop:ez2}]
    We wish to show that for any $1 \le t \le k+r,$ the sum of the largest $t$ elements of $\{a_1, \dots, a_k, c_1, \dots, c_r\}$ is at least the sum of the largest $t$ elements of $\{b_1, \dots, b_k, c_1, \dots, c_r\}$. To prove this, assume WLOG that $a_1 \ge \cdots \ge a_k,$ $b_1 \ge \cdots \ge b_k$, and $c_1 \ge \cdots \ge c_r$. Suppose the largest $t$ elements of $\{b_1, \dots, b_k, c_1, \dots, c_r\}$ are $b_1, \dots, b_g, c_1, \dots, c_{r-g}$, where $0 \le g \le t$. Then, since $\{a_1, \dots, a_k\}$ majorizes $\{b_1, \dots, b_k\}$, we have that $a_1 + \cdots + a_g \ge b_1 + \cdots + b_g$, so $a_1 + \cdots + b_g + c_1 + \cdots + c_{r-g} \ge b_1 + \cdots + b_g + c_1 + \cdots + c_{r-g}.$ So, the sum of the largest $t$ elements of $\{a_1, \dots, a_k, c_1, \dots, c_r\}$ must be at least as large as $b_1 + \cdots + b_g + c_1 + \cdots + c_{r-g}.$
\end{proof}

\begin{proof}[Proof of Proposition \ref{prop:perturbation}]
    We additionally show that in each step, that after each perturbation, none of the terms will be $x_1$ or less than $x_n$. We prove this by induction on $n$, the length of the sequences $\{x_i\}$ and $\{y_i\}$. If $n = 1$, then we must have $x_1 = y_1$ for majorization, so the proof is trivial. 
    
    For $n \ge 2,$ assume WLOG that $\{x_i\}$ and $\{y_i\}$ are sorted in nondecreasing order. Then, if $\{x_i\}$ dominates $\{y_i\}$, $x_1 \ge y_1 \ge y_n \ge x_n$. If $x_1 = y_1,$ then we use the inductive hypothesis on $\{x_2, \dots, x_n\}$ and $\{y_2, \dots, y_n\}$, and keep $x_1 = y_1$ as is. Since each perturbation does not increase anything to greater than $x_2 \le x_1$ or decrease anything to less than $x_n$, we have that the order of the full sequence (including the first term which remains at $y_1$) is preserved and no term will ever increase to greater than $x_1$ or less than $x_n$. A symmetric argument takes care of the case when $x_n = y_n$, where we use the inductive hypothesis on $\{x_1, \dots, x_{n-1}\}$ and $\{y_1, \dots, y_{n-1}\}$.
    
    Otherwise, let $\alpha = \min(x_1-y_1, y_n-x_n)$, and let $\eps = \frac{\alpha}{1 + \lfloor \alpha \rfloor}.$ For each of $1 + \lfloor \alpha \rfloor$ steps, increase $y_1$ by $\eps$ and decrease $y_n$ by $\eps$. Since $y_1 \ge y_2 \ge \cdots \ge y_n$ at the beginning, we clearly preserve this at each step, and since $x_1 \ge y_1$ and $x_n \ge y_n$ at the beginning, we also clearly preserve this throughout. Finally, we have either increased $y_1$ to $x_1$ or decreased $y_n$ to $x_n$, so we can apply the previous paragraph.
\end{proof}

\begin{proof}[Proof of Proposition \ref{prop:super_dom_ez}]
    Suppose the contrary, so there exists $1 \le j \le n$ such that $j$th largest element of $x$ is less than the $j$th largest element of $y$. Then, if $x^{(i)}$ represents the $i$th largest element of $\{x_i\}$ and $y^{(i)}$ represents the $i$th largest element of $\{y_i\}$, then $y^{(1)} \ge \cdots \ge y^{(j)} > x^{(j)} \ge \cdots \ge x^{(n)}.$ Note, however, that $y^{(1)}, \dots, y^{(j)}$ comprises $j$ numbers and $x^{(j)}, \dots, x^{(n)}$ comprises $n-j+1$ numbers. Therefore, by the Pigeonhole principle, there exists $1 \le k \le j$ and $j \le \ell \le n$ such that $y^{(k)} = y_m$ and $x^{(\ell)} = x_m$ for some $m \in [n]$. But then this implies that $y_m > x_m$, which is a contradiction.
\end{proof}

\end{document}